\documentclass[11pt,a4paper]{article}    % onecolumn (standardformat)

\usepackage[margin=1in]{geometry} %In Folder
\usepackage{amsmath,amssymb,amsthm}
\usepackage{color}
\usepackage{relsize}
\usepackage[perpage,symbol*]{footmisc}
\usepackage{authblk}
\usepackage{enumerate}
\usepackage{cases}
\usepackage{cite}

\newtheorem{theorem}{Theorem}[section]
\newtheorem{proposition}[theorem]{Proposition}
\newtheorem{lemma}[theorem]{Lemma}
\newtheorem{corollary}[theorem]{Corollary}

\newtheorem{remark}{Remark}

\newcommand{\al}{\alpha}
\newcommand{\bt}{\beta}

\newcommand{\be}{\begin{equation}}
\newcommand{\ee}{\end{equation}}
\newcommand{\bea}{\begin{eqnarray}}
\newcommand{\eea}{\end{eqnarray}}
\newcommand{\e}{{\rm e}}

\def\tr {{\rm tr}}

\numberwithin{equation}{section}

\title{Laguerre Unitary Ensembles with Jump Discontinuities, PDEs and the Coupled Painlev\'{e} V System}
\author[1]{Shulin Lyu}
\author[2]{Yang Chen}
\author[3,\thanks{Author to whom any correspondence should be addressed.}] {Shuai-Xia Xu}
\affil[1]{School of Mathematics and Statistics, Qilu University of Technology (Shandong Academy of Sciences), Jinan
250353, China}
\affil[2]{Department of Mathematics, Faculty of Science and Technology, University of Macau, Macau, China}
\affil[3]{Institut Franco-Chinois de l'Energie Nucl\'{e}aire, Sun Yat-sen University, Guangzhou 510275, China}

\date{\today}

\begin{document}
\maketitle
\begin{abstract}
We study the Hankel determinant generated by the Laguerre weight with jump discontinuities at $t_k, k=1,\cdots,m$. By employing the ladder operator approach to establish Riccati equations,  we show that $\sigma_n(t_1,\cdots,t_m)$, the logarithmic  derivative of the $n$-dimensional Hankel determinant, satisfies a generalization of the $\sigma$-from of Painlev\'{e} V equation. Through investigating the Riemann-Hilbert problem for the associated orthogonal polynomials and via Lax pair, we express $\sigma_n$ in terms of solutions of a coupled Painlev\'{e} V system. We also build relations between the auxiliary quantities introduced in the above two methods, which provides connections between the Riccati equations and Lax pair. In addition, when each $t_k$ tends to the hard edge of the spectrum and $n$ goes to $\infty$, the scaled $\sigma_n$ is shown to satisfy a generalized Painlev\'{e} III equation.
\end{abstract}

$\mathbf{Keywords}$: Laguerre unitary ensembles; Hankel determinant; Orthogonal polynomials;

Painlev\'{e} equations; Riemann-Hilbert problems

$\mathbf{Mathematics\:\: Subject\:\: Classification\:\: 2020}$: 33E17; 34M55; 41A60; 42C05

\tableofcontents

\section{Introduction and Statement of Main Results}
The $n$-dimensional Laguerre unitary ensemble (LUE for short) is a group of $n\times n$ Hermitian random matrices with eigenvalues having the following joint probability density function
\begin{align*}
p(x_1,x_2,\cdots,x_n)=\frac{1}{C_n}\cdot\frac{1}{n!}\prod_{1\leq i<j\leq n}(x_i-x_j)^2\prod_{k=1}^n x_k^{\alpha}\e^{-x_k},
\end{align*}
where $x_k\in[0,\infty)$ for $k=1,2,\cdots,n$ and $\al>-1$.
See \cite[sections 2.5, 2.6 and 3.3]{Mehta} for more details. The normalization constant $n!C_n$, which is also known as the partition function, has the following explicit expression \cite[equation (17.6.5)]{Mehta}
\begin{align*}
n!C_n:=&\int_{[0,\infty)^n} \prod_{1\leq i<j\leq n}(x_i-x_j)^2\prod_{k=1}^n x_k^{\alpha}\e^{-x_k}dx_k\\
=&\prod_{k=1}^n k!\cdot\Gamma(\alpha+k).
\end{align*}
%namely,
%\begin{align}\label{LUEconst}
%C_n=\prod_{k=1}^{n-1} k!\cdot\prod_{k=1}^n\Gamma(\alpha+k).
%\end{align}
For any interval $I\subset[0,\infty)$, the probability that all the eigenvalues of LUE lie in $I$ is given by
\begin{align}
P\left(\{x_k\in I:   k=1,\cdots,n\}\right) =&\frac{1}{C_n}\cdot\frac{1}{n!}\int_{I^n}\prod_{1\leq i<j\leq n}(x_i-x_j)^2\prod_{k=1}^n x_k^{\alpha}\e^{-x_k}dx_k\nonumber\\
=&\frac{1}{C_n}\det\left(\int_I x^{i+j}x^{\alpha}\e^{-x}dx\right)_{i,j=0}^{n-1},\label{pLUE}
\end{align}
where the second equality is due to Heine's formula (see e.g. \cite[sections 2.1 and 2.2]{Szego}) and the determinant is called the Hankel determinant.

In this paper, we are concerned with the Hankel determinant generated by the moments of the Laguerre weight multiplied by a factor with $m$ jump discontinuities, namely
\begin{align}\label{m-HD}
D_n\left(\,\vec{t}\,\right):=\det\left(\int_{0}^{\infty} x^{i+j}w(x;\vec{t}\,)dx\right)_{i,j=0}^{n-1},
\end{align}
where $\vec{t}=\left(t_1,\cdots,t_m\right)$ with $0<t_1<\cdots<t_m$, and the weight function reads
\begin{equation}\label{eq:weight}
w(x;\vec{t}\,):=x^{\al}\e^{-x}\left(\omega_0+\sum_{k=1}^m \omega_k\theta(x-t_k)\right),\qquad\qquad \al>-1.
\end{equation}
Here $\sum\limits_{k=0}^\ell \omega_k\geq0$ for $\ell=0,1,\cdots,m$,
and $\theta(x)$ is the step function which is 1 for $x>0$ and 0 otherwise.

When $m=2$, according to \eqref{pLUE} and \eqref{m-HD}, we find that the probability that the interval $(t_1,t_2)$ has no or all eigenvalues of LUE is given by $D_n(t_1,t_2)/C_n$ with $\omega_0=1,\omega_1=-1,\omega_2=1$ and $\omega_0=0,\omega_1=1,\omega_2=-1$, respectively. The latter probability was studied in \cite{BasorChenZhang} by using the ladder operator approach \cite{Ismail}, a formalism adapted to monic polynomials orthogonal with respect to $w(x;t_1,t_2)$, and the  logarithmic derivative of $D_n(t_1,t_2)$ was shown to satisfy a second order partial differential equation (PDE for short) which can be viewed as a two-variable generalization of the Painlev\'{e} V equation ($P_V$ for short).

For $w(x;\vec{t}\,)$, we can define the associated monic orthogonal polynomials by
\begin{align}\label{m-or}
\int_{0}^{\infty}P_n(x;\vec{t}\,)P_m(x;\vec{t}\,)w(x;\vec{t}\,)dx=h_{n}\delta_{mn},\qquad m,n\geq0,
\end{align}
where $\delta_{mn}=1$ for $m=n$ and 0 otherwise, and
\begin{align}\label{defPn}
P_n\left(x;\vec{t}\,\right):=x^n+p(n,\vec{t}\,)x^{n-1}+\cdots+P_n(0;\vec{t}\,).
\end{align}
From this definition, there follows the three-term recurrence relation
\begin{align}\label{OP-recu}
xP_n(x;\vec{t}\,)=P_{n+1}(x;\vec{t}\,)+\al_n(\vec{t}\,)P_n(x;\vec{t}\,)+\bt_n(\vec{t}\,)P_{n-1}(x;\vec{t}\,),\qquad n\geq0,
\end{align}
where $P_0(x;\vec{t}\,):=1$ and $\bt_0P_{-1}(x;\vec{t}\,):=0$, and the recurrence coefficients are given by
\begin{align}
\al_n(\vec{t}\,)=&p(n,\vec{t}\,)-p(n+1,\vec{t}\,),\qquad n\geq0,\label{m-alp}\\ \bt_n(\vec{t}\,)=&\frac{h_n(\vec{t}\,)}{h_{n-1}(\vec{t}\,)},\qquad n\geq1,\nonumber
\end{align}
with $p(0,\vec{t}\,):=0$.
%As an immediate consequence of \eqref{m-alp}, we have
%\[\sum_{k=0}^{n-1}\al_k(\vec{t}\,)=-p(n,\vec{t}\,).\]
The recurrence relation indicates the following Christoffel-Darboux formula
\[\sum_{k=0}^{n-1}\frac{P_k(x;\vec{t}\,)P_k(y;\vec{t}\,)}{h_k(\vec{t}\,)}=\frac{P_n(x;\vec{t}\,)P_{n-1}(y;\vec{t}\,)-P_{n-1}(x;\vec{t}\,)P_n(y;\vec{t}\,)}{h_{n-1}(\vec{t}\,)(x-y)}.\]
With all these properties, one can derive a pair of ladder operators satisfied by $P_n(z;\vec{t})$ and three compatibility conditions, based on which the ladder operator approach is formulated. We should point out that, given an arbitrary positive function whose moments of all orders exist, the associated monic orthogonal polynomials have the properties mentioned above. For more details, the reader may consult \cite[section 3.2]{Szego} and \cite{Ismail}.

The ladder operator approach was widely used to study problems related to unitary ensembles which involve one or two perturbation variables or jump discontinuites. For example, the Hankel determinant generated by the singularly perturbed Laguerre weight $x^{\al}\e^{-x-t/x}$, which arises from a finite-temperature integrable quantum field theory, was investigated in \cite{ChenIts10} for $x\in(0,\infty)$ and in \cite{LyuGriffinChen19} for $x\in(s,\infty)$;
The Hankel determinant of a two-time deformation of the Laguerre weight characterizes the multiple-input-multiple-output wireless communication system \cite{ChenHaqMcKay13}, and its log-derivative was shown to satisfy a two-variable generalization of the $\sigma$-form of $P_V$. Under double scaling, the log-derivative of the gap probability on $(-a,a)$ for Gaussian and symmetric Jacobi unitary ensembles were found to satisfy the $\sigma$-form of $P_{V}$\cite{LyuChenFan18, MinChen18}; etc.. For problems involving more than two variables, the ladder operator approach is rarely chosen, which may be caused due to the complexity and difficulty of the derivation of the coupled PDEs.

It is well-known that orthogonal polynomials satisfy a Riemann-Hilbert (RH for short) problem. This fact provides an alternative method to study unitary ensembles. By conducting Deift-Zhou steepest descent analysis \cite{Deift} to the RH problem, one can derive asymptotics for orthogonal polynomials and quantities relating unitary ensembles. For problems involving more than one variables, a coupled Painlev\'{e} system is usually established. See, for instance \cite{ACM, CD, DaiXuZhang18, DaiXuZhang19}.

The Hankel determinant generated by the Gaussian weight with $m$ jumps was studied in \cite{WuXu20}. When the dimension of the Hankel matrix is finite, by considering the RH problem for the associated monic orthogonal polynomials and with the help of the Lax pair, a coupled $P_{IV}$ system was derived.
As the jump discontinuities tend to the edge of the spectrum and the dimension becomes large, by carrying out Deift-Zhou steepest descent analysis to the RH problem mentioned above, an asymptotic expression for the orthogonal polynomials in terms of quantities satisfying a coupled $P_{II}$ system was established.
The two-jump case ($m=2$) was investigated in \cite{MinChen19} via the ladder operator approach and a second order PDE was deduced for the  logarithmic derivative of the Hankel determinant.
By making use of the finite dimensional results therein, the first and the second author of the present paper reproduce the aforementioned coupled $P_{IV}$ and $P_{II}$ system of \cite{WuXu20} with $m=2$ \cite{LyuChen20}.
See also \cite{cc,CD,XuDai} for the applications of coupled $P_{II}$ system in the studies Airy kernel determinants with discontinuities and the Painlev\'e-type kernel determinants.
%This provides a possible insight into the relationship between the ladder operators and the Lax pair of the RH problem for orthogonal polynomials.

In this paper, we adopt both approaches mentioned above to study $D_n\left(\vec{t}\right)$. We establish a coupled $P_V$ system and a second order PDE for the  logarithmic derivative of $D_n(\vec{t})$. By making comparison and connection between the results obtained by these two different methods, we find that the ladder operators are closely related to the Lax pair of the RH problem for orthogonal polynomials.
We also consider the behavior of $D_n(\vec{t})$ under the double scaling that $\vec{t}\rightarrow\vec{0}$ and $n\rightarrow\infty$ such that
$\vec{t} =\frac{1}{4n} \vec{s}$, with $\vec{s}=(s_1,s_2,\cdots, s_m)$  fixed. We deduce the limiting PDE satisfied by the  logarithmic derivative of the double scaled Hankel determinant, which can be viewed as an $m$-variable generalization of the $\sigma$-form of $P_{III}$.  If further $s_k=\nu_k\tau$ with $\nu_k$ arbitrary given constant, we obtain the expression for the scaled Hankel determinant in terms of quantities that satisfy a coupled ODEs in variable $\tau$ and whose square roots were shown in \cite{cd} to be approximated by the Bessel function.

The main results of this paper are presented as below.

\subsection{Statement of main results}
%\subsubsection{Symbols and notations}
%We introduce the following $2m$ auxiliary quantities
%\begin{align*}
%R_{n,k}(\vec{t}):=&\omega_k\frac{w_0(t_k)}{h_n(\vec{t})}P_n^2(t_k;\vec{t}),\\
%r_{n,k}(\vec{t}):=&\omega_k\frac{w_0(t_k)}{h_{n-1}(\vec{t})}P_n(t_k;\vec{t})P_{n-1}(t_k;\vec{t}),
%\end{align*}
%for $k=1,\cdots,m$.

\subsubsection{Results for finite dimension $n$: PDEs and the coupled $P_{V}$ system}
Define
\begin{equation}\label{eq:sigman}
\sigma_n(\vec{t}):=\delta\,{\ln} D_n(\vec{t}), \qquad \delta:=\sum_{k=1}^m t_k\partial_{t_k},
\end{equation}
where $\partial_{t_k}$ denotes $\frac{\partial}{\partial t_k}$. For ease of notations, we also represent $\frac{\partial^2}{\partial t_k^2}$ and $\frac{\partial^2}{\partial t_j\partial t_k} (j\neq k)$ by $\partial_{t_kt_k}^2$ and $\partial_{t_jt_k}^2$ respectively.

To take the ladder operator approach, we introduce $2m$ auxiliary quantities allied to $P_n(x;\vec{t})$:

\begin{align}\label{eq:Rn}
R_{n,k}(\vec{t}):=&\omega_k\frac{w_0(t_k)}{h_n(\vec{t})}P_n^2(t_k;\vec{t}),\\ \label{eq:rn}
r_{n,k}(\vec{t}):=&\omega_k\frac{w_0(t_k)}{h_{n-1}(\vec{t})}P_n(t_k;\vec{t})P_{n-1}(t_k;\vec{t}),
\end{align}
with $k=1,2,\cdots,m$ and $w_0(x)=x^{\al}\e^{-x}$.
We first derive the Riccati equations for $\{R_{n,k}\}$ and $\{r_{n,k}\}$, which give us a system of PDEs satisfied by $\{R_{n,k}\}$. Then we establish the expression for $\sigma_n$ in terms of $\{R_{n,k}\}$ and $\{r_{n,k}\}$, from which a second order PDE satisfied by $\sigma_n$ is deduced.
\begin{theorem}\label{LO-result}
$R_{n,k}(\vec{t})$ and $r_{n,k}(\vec{t})$, defined in \eqref{eq:Rn} and \eqref{eq:rn}, respectively,  satisfy the following analogs of Riccati equations
\begin{subequations}\label{m-Ric}
\begin{align}
\delta R_{n,k}=&2r_{n,k}+\Bigg(\sum_{j=1}^m t_j R_{n,j}-t_k+2n+\al\Bigg)R_{n,k},\label{m-Ric1}\\
\delta r_{n,k}=&\frac{r_{n,k}^2}{R_{n,k}}-R_{n,k}\left(\dfrac{\bigg(n+\al+\sum\limits_{j=1}^m r_{n,j}\bigg)\bigg(n+\sum\limits_{j=1}^m r_{n,j}\bigg)}{1-\sum\limits_{j=1}^m R_{n,j}}+\sum\limits_{j=1}^m\frac{r_{n,j}^2}{R_{n,j}}\right),\label{m-Ric2}
\end{align}
\end{subequations}
for $k=1,\cdots,m$, from which we establish a coupled second order PDEs for $\{R_{n,k}(\vec{t})\}:$
\begin{equation}\label{m-RPDE}
\begin{aligned}
&\delta^2 R_{n,k}+\bigg(t_k-\sum_{j\neq k} t_j R_{n,j}-2t_k R_{n,k}-2n-\al\bigg)\delta R_{n,k}-R_{n,k}\sum_{j\neq k} t_j\cdot \delta R_{n,j}\\
=&\left(\frac{\bigg(2(n+\al)+\sum\limits_{j=1}^m 2r_{n,j}\bigg)\bigg(2n+\sum\limits_{j=1}^m 2r_{n,j}\bigg)}{2\bigg(\sum\limits_{j=1}^m R_{n,j}-1\bigg)}-\sum\limits_{j=1}^m\frac{(2r_{n,j})^2}{2R_{n,j}}+\sum_{j=1}^m t_j R_{n,j}-t_k\right)R_{n,k}+\frac{(2r_{n,k})^2}{2R_{n,k}},
\end{aligned}
\end{equation}
with $k=1,\cdots,m$, where $2r_{n,k}$ is given by \eqref{m-Ric1} and
\begin{align*}
\delta^2=&\sum_{k=1}^m t_k^2\partial_{t_kt_k}^2+2\sum_{1\leq j<k\leq m} t_j t_k\partial_{t_j t_k}^2+\delta.
\end{align*}
Moreover, $\sigma_n(\vec{t})$, defined in \eqref{eq:sigman},  is expressed in terms of $\{R_{n,k}(\vec{t})\}$ and $\{r_{n,k}(\vec{t})\}$ by
\begin{align}\label{sigmarR}
\sigma_n=\frac{\left(n+\al+\sum\limits_{k=1}^m r_{n,k}\right)\left(n+\sum\limits_{k=1}^m r_{n,k}\right)}{\sum\limits_{k=1}^m R_{n,k}-1}-\sum\limits_{k=1}^m\frac{r_{n,k}^2}{R_{n,k}}+\sum_{k=1}^m t_k r_{n,k}+n(n+\al),
\end{align}
and $\sigma_n(\vec{t})$ satisfies the following PDE
\begin{equation}\label{m-sigmaPDE}
\begin{aligned}
&\bigg(\sum\limits_{k=1}^m{\rm sgn}(\omega_k)\sqrt{\Delta_k(\vec{t})}-2\left(\delta\sigma_n-\sigma_n+n(n+\al)\right)\bigg)^2-\bigg(\sum_{k=1}^m \sum_{j=1}^m t_j\cdot\partial_{t_kt_j}^2\sigma_n\bigg)^2\\
&=4\left(\delta\sigma_n-\sigma_n+n(n+\al)\right)\bigg(n+\al+\sum\limits_{k=1}^m \partial_{t_k}\sigma_n\bigg)\bigg(n+\sum\limits_{k=1}^m \partial_{t_k}\sigma_n\bigg),
\end{aligned}
\end{equation}
where ${\rm sgn}(\omega_k)$ is the sign function of $\omega_k$, which is $1$ for $\omega_k>0$, $-1$ for $\omega_k<0$ and $0$ for $\omega_k=0$. Here
\begin{align*}
\Delta_k(\vec{t}):=&\bigg(\sum_{j=1}^m t_j\cdot\partial_{t_kt_j}^2\sigma_n\bigg)^2+4\left(\delta\sigma_n-\sigma_n+n(n+\al)\right)\left(\partial_{t_k}\sigma_n\right)^2.
\end{align*}
\end{theorem}
\begin{remark}
Supposing $\vec{t}=\eta(b_1,\cdots,b_m)=:\eta\vec{b}$,  $\eta\in\mathbb{R}$, we have $\delta=\eta\frac{d}{d\eta}$ and \eqref{m-RPDE} is reduced to a second order ODE satisfied by $R_{n,k}$, reading
\begin{equation*}
\begin{aligned}
&\eta^2R_{n,k}''+\eta\bigg(\eta\Big(b_k-\sum_{j\neq k} b_j R_{n,j}-2b_k R_{n,k}\Big)-2n+1-\al\bigg)R_{n,k}'-\eta^2R_{n,k}\sum_{j\neq k} b_jR_{n,j}'\\
=&\left(\frac{\bigg(2(n+\al)+\sum\limits_{j=1}^m 2r_{n,j}\bigg)\bigg(2n+\sum\limits_{j=1}^m 2r_{n,j}\bigg)}{2\bigg(\sum\limits_{j=1}^m R_{n,j}-1\bigg)}-\sum\limits_{j=1}^m\frac{(2r_{n,j})^2}{2R_{n,j}}+\eta\sum_{j=1}^m b_j R_{n,j}-\eta b_k\right)R_{n,k}+\frac{(2r_{n,k})^2}{2R_{n,k}},
\end{aligned}
\end{equation*}
for $k=1,\cdots,m$, where \[2r_{n,k}=\eta R_{n,k}'-\Bigg(\eta\sum_{j=1}^m b_j R_{n,j}-\eta b_k+2n+\al\Bigg)R_{n,k}.\]
\end{remark}

\begin{remark}
When $m=1$, with $t_1$ replaced by $t$ and $R_{n,1}$ by $R_n$, equation \eqref{m-RPDE} is reduced to
\begin{equation}\label{1-Rn-eq}
\begin{aligned}
R_{n}''
=&\frac{2R_n-1}{2R_n(R_n-1)}\left(R_n'\right)^2-\frac{R_n'}{t}+R_n^3+\left(\frac{2n+1+\alpha}{t}-\frac{3}{2}\right)R_n^2\\
&+\left(-\frac{2n+1+\alpha}{t}+\frac{1}{2}\right)R_n-\frac{\alpha^2}{2t^2}\cdot\frac{R_n}{R_n-1},
\end{aligned}
\end{equation}
which can be transformed into $P_V\left(0,-\alpha^2/2,2n+1+\alpha,-1/2\right)$ satisfied by $1-1/R_n(t)$.
Putting $m=1$ in \eqref{m-sigmaPDE}, we get
\begin{align}\label{1-sg-eq}
(t\sigma_n'')^2=\left((2n+\al-t)\sigma_n'+\sigma_n\right)^2-4\left(t\sigma_n'-\sigma_n+n(n+\al)\right)(\sigma_n')^2,
\end{align}
which is the Jimbo-Miwa-Okamoto (J-M-O for short) $\sigma$-form of  the above-mentioned $P_V$.
Equations \eqref{1-Rn-eq} and \eqref{1-sg-eq} agree with (4.8) and (4.11) of \cite{BasorChen09}, respectively.

With these discussions, we may treat \eqref{m-RPDE} with $R_{n,k}$ replaced by $1/(1-R_{n,k})$ and \eqref{m-sigmaPDE} as an $m$-variable generalization of $P_V$ and the J-M-O $\sigma$-form of $P_V$, respectively.
\end{remark}

Denote by $H_V$ the Hamiltonian for $P_V\left(\frac{n^2}{2},-\frac{(n+\al)^2}{2},\al+1,-\frac{1}{2}\right)$, i.e.
\begin{equation}\label{eq:H-v}
  tH_{V}(t, u,v)=u^2v(v-1)^2+tuv+(n+\al)u(v-1)^2-\al uv(v-1).
\end{equation}
See \cite[equations (0.4) and (0.5) ]{Okamoto} and \cite[equations (2.1)--(2.5)]{FW}. By establishing the RH problem for $P_n(x;\vec{t})$ and with the aid of Lax pair, we show that $\sigma_n(\vec{t})$ can be expressed in terms of a coupled $P_V$ system.
\begin{theorem}\label{cpv}
$\sigma_n(\vec{t})$ is connected with the Hamiltonian \eqref{eq:H-v} by
  \begin{equation}\label{eq:H}
\sigma_n(\vec{t})=\sum_{k=1}^m t_kH_{V}(t_k, u_k, v_k)  +\frac{1}{2}\sum_{ j\neq k}u_ju_k(v_j+v_k)(v_j-1)(v_k-1),
\end{equation}
where $u_k(\vec{t})$ and $v_k(\vec{t}), k=1,\cdots,m$, satisfy the coupled $P_V$ system
\begin{equation}\label{m-eq:cpv}
 \left\{ \begin{array}{l}
       %          t\frac{d u_1}{dt}=-a_1t u_1-u_1^2(v_1-1)(3v_1-1)-u_1u_2(v_2-1)(2v_1+v_2-1)-2n
%u_1v_1+(2n+\alpha) u_1,\\
                   \delta u_k=-t_k u_k-u_k\sum\limits_{j=1}^mu_j(v_j-1)(2v_k+v_j-1)-2nu_kv_k+(2n+\alpha)  u_k, \\
         %          t\frac{d v_1}{dt}=a_1t v_1+2u_1v_1(v_1-1)^2+u_2(v_1+v_2)(v_1-1)(v_2-1)-\frac{\alpha}{2}(v_1^2-1)+(n+\frac{\alpha}{2})(v_1-1)^2,\\
                     \delta v_k=t_k v_k+(v_k-1)\sum\limits_{j= 1}^mu_j(v_j-1)(v_k+v_j)+(n+\al)(v_k-1)^2-\al v_k(v_k-1),
\end{array}\right.
\end{equation}
which is equivalent to the following Hamiltonian formulation
\begin{equation}\label{eq:H-equation}
  \delta v_k=\frac{\partial \sigma_n}{\partial u_k} ,\qquad\qquad \delta u_k=-\frac{\partial \sigma_n}{\partial v_k}, \qquad\qquad k=1,\cdots,m.
\end{equation}
\end{theorem}

\begin{remark}
Setting $m=1$ in \eqref{m-eq:cpv} and replacing $t_1$ by $t$, we have
\begin{equation*}
 \Bigg\{ \begin{array}{l}
                   tu_1'=-t u_1-u_1^2(v_1-1)(3v_1-1)-2nu_1v_1+(2n+\alpha)  u_1, \\
 tv_1'=t v_1+2u_1v_1(v_1-1)^2+\frac{\alpha}{2}(1-v_1^2)+(n+\frac{\alpha}{2}) (v_1-1)^2.
\end{array}\Bigg.
\end{equation*}
Solving $u_1$ from the second equation and substituting it into the first one, we obtain
\begin{align*}
v_1''=&\left(\frac{1}{2v_1}+\frac{1}{v_1-1}\right)(v_1')^2-\frac{v_1'}{t}+\frac{(v_1-1)^2}{t^2}\left(\frac{n^2}{2}v_1-\frac{(n+\al)^2}{2v_1}\right)\\
&-(\al+1)\frac{v_1}{t}-\frac{1}{2}\frac{v_1(v_1+1)}{v_1-1},
\end{align*}
which is $P_V\left(\frac{n^2}{2},-\frac{(n+\al)^2}{2},\al+1,-\frac{1}{2}\right)$.
\end{remark}

Combining the results obtained by the ladder operator approach and by solving the RH problem for $P_n(x;\vec{t})$, we build the relationship between $\{R_{n,k},r_{n,k}\}$ which satisfy the $m$-variable generalization of $P_V$ and $\{u_k,v_k\}$ which are solutions of the coupled $P_V$ system.
In addition, we obtain expressions for a few quantities allied with $P_n(x;\vec{t})$ in terms of both $\{R_{n,k}, r_{n,k}\}$ and $\{u_k, v_k\}$, including the recurrence coefficients $\al_n(\vec{t})$ and $\bt_n(\vec{t})$, the $L^2$ norm of the monic orthogonal polynomial and its derivative with respect to $t_k$. To continue, we define $b_1(\vec{t})$ and $b_2(\vec{t})$ by
\begin{subequations}\label{eq:b-k}
\begin{align}
 b_1(\vec{t}):=&- \sum\limits_{k=1}^mu_k(v_k-1)-n,\label{defb1}\\
 b_2(\vec{t}):=&- \sum\limits_{k=1}^mu_kv_k(v_k-1)+n+\alpha.\label{defb2}
 \end{align}
\end{subequations}
\vspace{-6mm}
\begin{theorem}\label{LO-RHP}
For $k=1,\cdots, m$, $R_{n,k}(\vec{t})$ and $r_{n,k}(\vec{t})$ defined in \eqref{eq:Rn} and \eqref{eq:rn}  are connected with $u_k(\vec{t})$ and $v_k(\vec{t})$ by
\begin{subequations}\label{Rruv}
\begin{align}
                  R_{n,k}=&\frac{u_k}{b_1},\label{Rub}\\
                  r_{n,k}=&u_kv_k,\label{ruv}
\end{align}
\end{subequations}
or equivalently,
\begin{subequations}\label{uvRr}
\begin{align}
u_k=&R_{n,k}\cdot\dfrac{\sum\limits_{j=1}^m r_{n,j}+n}{\sum\limits_{j=1}^m R_{n,j}-1},\label{uRr}\\
v_k=&\dfrac{r_{n,k}}{R_{n,k}}\cdot\frac{\sum\limits_{j=1}^m R_{n,j}-1}{\sum\limits_{j=1}^m r_{n,j}+n}.\label{vRr}
\end{align}
\end{subequations}
And we have the following expressions
\begin{subequations}\label{alRu}
\begin{align}
 \al_n(\vec{t})=&\sum\limits_{k=1}^m t_k R_{n,k}+2n+1+\al\\
 =&\frac{1}{ b_1}\sum\limits_{k=1}^m t_k u_k+2n+1+\al,
\end{align}
\end{subequations}
\begin{subequations}\label{btRrb}
\begin{align}
\beta_{n}(\vec{t})=&\dfrac{\left(n+\al+\sum\limits_{k=1}^m r_{n,k}\right)\cdot\left(n+\sum\limits_{k=1}^m r_{n,k}\right)}{1-\sum\limits_{k=1}^m R_{n,k}}+\sum\limits_{k=1}^m\frac{r_{n,k}^2}{R_{n,k}}\label{btrR}\\
=&-b_1b_2,\label{btb1b2}
\end{align}
\end{subequations}
\begin{align}
h_{n}(\vec{t})=&-2\pi it_1^{2n+\alpha}b_1y,\qquad\qquad\label{hnt1b1y}\\
                     % h_{n-1}=2\pi is_1^{2n+\alpha}b_2^{-1}y,\\
\partial_{t_k}\ln h_n(\vec{t})=&-R_{n,k}=-\frac{u_k}{b_1},\qquad\qquad\label{m-DhR}
\end{align}
where $y(\vec{t})$ is connected with $u_k(\vec{t})$ and $v_k(\vec{t})$ by $\delta \ln y=\sum\limits_{k=1}^mu_k(v_k-1)^2-2n-\alpha.$
\end{theorem}

\begin{remark}
From the coupled $P_V$ system \eqref{m-eq:cpv}, we can derive the Riccati equations \eqref{m-Ric} which give rise to the second order PDEs satisfied by $R_{n,k}$. In the following derivation, we always bear \eqref{Rruv} in mind.

Indeed, differentiating both sides of \eqref{Rub}, we have
\begin{align}
\delta R_{n,k}=&\frac{\delta u_k}{b_1}-\frac{u_k\cdot\delta b_1}{b_1^2}.\label{Ric1-1}
\end{align}
Taking the derivative of both sides of \eqref{defb1}, and using \eqref{m-eq:cpv} to replace $\delta u_k$ and $\delta v_k$ in the resulting expression, we find
\begin{align*}
\delta b_1=-\sum_{k=1}^m t_k u_k-b_1\sum_{k=1}^m u_k(v_k-1)^2.
\end{align*}
Inserting it and $\delta u_k$ given by the first equation of \eqref{m-eq:cpv} into \eqref{Ric1-1}, we come to the Riccati equation \eqref{m-Ric1}.

To derive \eqref{m-Ric2}, we apply $\delta$ to \eqref{ruv} and get in view of \eqref{m-eq:cpv}
\begin{align}
\delta r_{n,k}=&\delta u_k\cdot v_k+u_k\cdot \delta v_k=\frac{r_{n,k}^2}{R_{n,k}}+b_1b_2 R_{n,k}.\label{m-Ric2-1-1}
\end{align}
From the definitions of $b_1$ and $b_2$, it follows that
\begin{align}
%b_1=&\frac{\sum_{k=1}^m r_{n,k}+n}{\sum_{k=1}^m R_{n,k}-1},\nonumber\\
b_1b_2=&\dfrac{\left(n+\al+\sum\limits_{k=1}^m r_{n,k}\right)\cdot\left(n+\sum\limits_{k=1}^m r_{n,k}\right)}{\sum\limits_{k=1}^m R_{n,k}-1}-\sum_{k=1}^m\frac{r_{n,k}^2}{R_{n,k}}.\label{b1b2}
\end{align}
Combining it with \eqref{m-Ric2-1-1} gives us \eqref{m-Ric2}.
\end{remark}
\subsubsection{Double scaling analysis at the hard edge: generalized $P_{III}$}
We consider the hard edge scaling limit of the Hankel determinant $D_n(\vec{t})$ under the assumption that $n\rightarrow\infty$ and $t_k\rightarrow0^+$ such that
\begin{align*}
s_k=4n t_k,\qquad k=1,2,\cdots, m,
\end{align*}
are fixed. This scaling is motivated by the following fact: if $t_j=0$ for $j=1,\cdots,m-1$ while $t_m\neq0$, then our weight function has only one jump, i.e. $x^{\al}\e^{-x}\left(\sum\limits_{k=0}^{m-1}\omega_k+\omega_m\theta(x-t_m)\right)$,
and its associated Hankel determinant with the special cases $\sum\limits_{k=0}^{m-1}\omega_k=0, \omega_m=1$ and $\sum\limits_{k=0}^{m-1}\omega_k=1, \omega_m=-1$ correspond to the smallest and largest eigenvalue distribution of LUE respectively, both of which under the scaling that $s_m=4nt_m$ tend to the Fredholm determinant of the Bessel kernel.

Assuming
\begin{equation}\label{eq:scalingt}
\vec{t}=\frac{1}{4n}\vec{s}, \quad \vec{s}=(s_1,s_2,\cdots, s_m),
\end{equation}
with $0<s_1<\cdots<s_m$, we have
\[\delta=\sum_{k=1}^m t_k\partial_{t_k}=\sum_{k=1}^m s_k\partial_{s_k}.\]
Define
 \begin{equation}\label{eq:sigmalimit}
%R_k(\vec{s}):=&\lim\limits_{n\rightarrow\infty}R_{n,k}\left(\frac{\vec{s}}{4n}\right),\label{m-limRnrnsigma-2}\\
%r_k(\vec{s})=&\lim\limits_{n\rightarrow\infty}\frac{r_{n,k}\left(\frac{\vec{s}}{4n}\right)}{n}\\
\sigma(\vec{s}):=\lim\limits_{n\rightarrow\infty}\sigma_n\left(\frac{\vec{s}}{4n}\right)=\delta\ln\lim\limits_{n\rightarrow\infty}D_n\left(\frac{\vec{s}}{4n}\right).
\end{equation}
By using the finite-$n$ results obtained via the ladder operator approach, we build direct relationships between the scaled $R_{n,k}, r_{n,k}$ and $\sigma(\vec{s})$. With them and the results presented in Theorem \ref{LO-result}, we establish the second order PDE satisfied by $\sigma(\vec{s})$.
\begin{proposition}\label{scaledRr}
Under the assumption \eqref{eq:scalingt}, the
uxiliary quantities $R_{n,k}(\vec{t})$ and $r_{n,k}(\vec{t})$ are scaled as below:
\begin{align}
-\lim\limits_{n\rightarrow\infty}\frac{r_{n,k}\left(\frac{\vec{s}}{4n}\right)}{n}
=&-4\partial_{s_k}\sigma(\vec{s})\label{m-limRnrnsigma-1}\\
=&\lim\limits_{n\rightarrow\infty}R_{n,k}\left(\frac{\vec{s}}{4n}\right)=:R_k(\vec{s}).\label{m-limRnrnsigma-2}
\end{align}
\end{proposition}
\begin{theorem}\label{m-resultscaled}
%$R_k(\vec{s})$ and $r_k(\vec{s})$ is connected with $\partial_{s_k}\sigma(\vec{s})$ by
%{\color{red}\begin{align}
%-r_k(\vec{s})=-4\partial_{s_k}\sigma(\vec{s})=R_k(\vec{s}).\label{m-limRnrnsigma-2}
%\end{align}
%$R_k(\vec{s})$ and $r_k(\vec{s})$ satisfy the following coupled Riccati equations}
Under the assumption \eqref{eq:scalingt}, the function $\sigma(\vec{s})$ specified by \eqref{eq:sigmalimit}
 can be expressed in terms of $R_k(\vec{s})$ and $\delta R_k(\vec{s})$ by
\begin{align}\label{m-sgRk-lim}
4\sigma=\frac{\left(\sum\limits_{k=1}^m \delta R_{k}\right)^2-\al^2}{\sum\limits_{k=1}^m R_{k}-1}-\sum_{k=1}^m \frac{\left(\delta R_k\right)^2}{R_k}-\sum_{k=1}^m s_kR_k-\al^2,
\end{align}
and $\{R_k(\vec{s})\}$ satisfy the following system of PDEs
\begin{equation}\label{m-RPDE-lim}
\begin{aligned}
&2\left(\delta^2 -\al\delta\right) R_k+\left(1-\frac{1}{R_k}\right)\left(\delta R_k-\al R_k\right)^2+R_k\cdot\sum_{j\neq k}\frac{\left(\delta R_j-\al R_j\right)^2}{R_j}\\
&+\frac{R_k}{1-\sum\limits_{j=1}^m{R_j}}\cdot\sum_{j=1}^m \left(\delta R_j-\al R_j \right)\cdot\bigg(\sum_{j=1}^m \left(\delta R_j-\al R_j \right)+2\al\bigg)\\
&+R_k\bigg(s_k-\sum_{j=1}^m s_j R_j\bigg)=0,\qquad\qquad k=1,\cdots,m,
\end{aligned}
\end{equation}
where
\begin{align*}
\delta^2=&\sum_{k=1}^m s_k^2\partial_{s_ks_k}^2+2\sum_{1\leq j<k\leq m} s_js_k\partial_{s_js_k}^2+\delta.
\end{align*}
Moreover, $\sigma(\vec{s})$ satisfies the following PDE
\begin{equation}\label{m-PDEsigma}
\begin{aligned}
\frac{-4\left[ \sum\limits_{k=1}^ms_k\partial_{s_ks_k}^2\sigma+\sum\limits_{1\leq j<k\leq m}(s_j+s_k)\partial_{s_js_k}^2\sigma\right]^2+\frac{\al^2}{4}}{4\sum\limits_{k=1}^m \partial_{s_k}\sigma+1}&\\
+\sum_{k=1}^m \bigg[\frac{1}{\partial_{s_k}\sigma}\bigg(\sum\limits_{j=1}^m s_j\partial_{s_js_k}^2\sigma\bigg)^2\bigg]&
+\delta\sigma-\sigma-\frac{\al^2}{4}=0.
\end{aligned}
\end{equation}
\end{theorem}
\begin{remark}
When $m=1$ and with $s_1$ replaced by $s$, expression \eqref{m-sgRk-lim} becomes
\begin{align}\label{Rrsigma}
4\sigma=\frac{\left(sR_1'\right)^2}{R_1(R_1-1)}-\frac{\al^2R_1}{R_1-1}-sR_1.
\end{align}
Equations \eqref{m-RPDE-lim} are reduced to
\begin{align}\label{ODER1}
R_1''=\frac{(2 R_1-1) (R_1')^2}{2 (R_1-1)
   R_1}-\frac{R_1'}{s}-\frac{\alpha ^2 R_1}{2 s^2
   (R_1-1)}+\frac{R_1}{2 s}(R_1-1),
\end{align}
which can be transformed into $P_V\left(0,-\alpha^2/2,1,0\right)$ satisfied by
\[Y(z):=1-1/R_1(s), \qquad\qquad \text{~with~} z=\frac{s}{2}.\]
%namely,
%\begin{align}\label{SPV}
%Y''=\left(\frac{1}{2 Y}+\frac{1}{Y-1}\right)
%   (Y')^2-\frac{Y'}{z}-\frac{\alpha ^2 (Y-1)^2}{2
%   z^2 Y}+\frac{Y}{z}.
%\end{align}
According to Theorem 3.9 of \cite{Assche}, this $P_V$ can be converted into $P_{III}(-2(\sqrt{\alpha^2}+1),2\sqrt{\alpha^2},1,-1)$ satisfied by
\[\nu(x)=\frac{\sqrt{2z}Y}{z\frac{d}{dz}Y+\sqrt{\alpha^2}Y-\sqrt{\alpha^2}}
=\frac{\sqrt{s}R_1(R_1-1)}{s\frac{d}{ds}R_1-\sqrt{\alpha^2}R_1},\qquad\qquad x^2=2z=s.\]
In this sense, \eqref{m-RPDE-lim} can be viewed as an $m$-variable generalization of $P_{III}$.
To continue, \eqref{m-PDEsigma} now reads
\begin{align}\label{sigmaODEB2=0}
(s\sigma'')^2=&\sigma'(4\sigma'+1)(\sigma-s\sigma')+\alpha^2(\sigma')^2,
\end{align}
which is identified to be the J-M-O $\sigma$ form of $P_{III}$ satisfied by $\sigma(-4\zeta)$ with $\zeta=-s/4$ (see \cite[formula (3.13)]{Jimbo} with the parameters
$\theta_0=\alpha,~\theta_{\infty}=-\alpha$ therein). This result coincides with the one in \cite{BasorChen09}.

In particular, for $m=1$ with $\omega_0=0$ and $\omega_1=1$, we know that $D_n(t_1)$ represents the smallest eigenvalue distribution of LUE up to a constant term. The double scaling analysis of this case was carried out by Tracy and Widom \cite{TW161} where an auxiliary quantity $q(s)$ was introduced.
%We would like to check the consistency between their results and ours.
By combining equation (2.26) and the one below it in \cite{TW161} with the definition of $\sigma(s)$ and \eqref{Rrsigma} in our paper, we find that $\sigma(s), R(s)$ and $q^2(s)$ in \cite{TW161} correspond to our $-\sigma(s),-\sigma(s)/s$ and $R_1(s)$ respectively.
We now check the consistency between their results and our equations \eqref{Rrsigma}--\eqref{sigmaODEB2=0}.
In fact, combining (1.20) of \cite{TW161} with $q(s)=cos \psi(s)$ therein, we obtain an expression that agrees with \eqref{Rrsigma}.
Substituting $R_1(s)=q^2(s)$ into \eqref{ODER1}, we are led to (1.16) of \cite{TW161}. (1.21) of \cite{TW161} is readily seen to be in accord with \eqref{sigmaODEB2=0}. The last thing to remark is that, as was pointed out in \cite{TW161}, the quantity
\[y(x):=\frac{q(s)-1}{q(s)+1},\qquad\qquad \text{~with~} x=\sqrt{s},\]
satisfies $P_V\left(\alpha^2/8,-\alpha^2/8,0,-2\right)$.
\end{remark}

%\item [(b)] {\bf with one scaled variable -- {\color{red}integral representation of scaled Hankel determinant}}\\
To continue, we let $s_i=\nu_i\tau, \tau>0$. Then \[\vec{t}=(t_1,t_2,\cdots,t_m)=\frac{\vec{s}}{4n}=\frac{\tau}{4n}\vec{\nu}, \qquad\vec{\nu}=(\nu_1, \cdots, \nu_m),\]
and we now have only one scaled variable $\tau$. Supposing $c_{k}=\sum\limits_{j=0}^{k-1} \omega_j/\omega, k=1,\cdots, m$, with $\omega_k$ positive and $\omega=\sum\limits_{k=0}^m\omega_k$ strictly positive, we find
 \begin{align*}
\frac{D_n(\vec{t})}{D_n(\vec{0})}&=\frac{\det\left(\displaystyle{\int_0^{+\infty}\bigg(1-\sum\limits_{k=1}^m(1-c_k)\chi_{(t_{k-1},t_k)}(x)\bigg)x^{i+j}x^{\alpha}{\e}^{-x}dx}\right)_{i,j=0}^{n-1}}{\det\left(\displaystyle{\int}_0^{+\infty}x^{i+j}x^{\alpha}{\e}^{-x}dx\right)_{i,j=0}^{n-1}}\nonumber\\
&=\det\big(I-\chi_{(0,t_m)}\sum_{k=1}^m(1-c_k)K_n\chi_{(t_{k-1},t_{k})}\big),
 \end{align*}
where  $t_0=0$ and
\begin{align*}
K_n(x,y)=(xy)^{\alpha/2}e^{-(x+y)}\sum_{j=0}^{n-1}L_j(x)L_j(y),
\end{align*}
with $L_j(x)$ being the normalized Laguerre polynomial of degree $j$.
It is known that for bounded $x,y\in(0,+\infty)$
 \begin{equation*}
\lim_{n\to\infty} \frac{1}{4n}K_n\left(\frac{x}{4n},\frac{y}{4n}\right)=K_{Bes}(x,y),
\end{equation*}
where $K_{\text{Bes}}(x,y)$ is the Bessel kernel
 \begin{equation*}
K_{\text {Bes}}(x,y)=\frac{J_{\alpha}(\sqrt{x})\sqrt{y}J_{\alpha}'(\sqrt{y})-\sqrt{x}J_{\alpha}'(\sqrt{x})J_{\alpha}(\sqrt{y})}{2(x-y)},
\end{equation*}
with  $J_{\alpha}$ denoting the Bessel function of the first kind of order $\alpha$.
This limit implies
 \begin{equation*}
\lim_{n\to\infty}\frac{D_n(\vec{t})}{D_n(\vec{0})}=\det\bigg(I-\chi_{(0,\nu_m\tau)}\sum_{j=1}^m(1-c_j)K_{\text{Bes}}\chi_{(\nu_{j-1}\tau,\nu_{j}\tau)}\bigg). \end{equation*}
According to the results presented in Theorem \ref{m-resultscaled}, we obtain the integral representation of $\lim\limits_{n\to\infty}\frac{D_n(\vec{t})}{D_n(\vec{0})}$, or equivalently, the determinant of the Bessel kernel,  in terms of $R_k(\vec{s})$.
\begin{corollary}\label{tau-Rq}
For $\vec{t}=\frac{\vec{s}}{4n}=\frac{\tau}{4n}\vec{\nu}$, $\vec{\nu}=(\nu_1, \cdots, \nu_m)$, define
 \begin{align}\label{defRhat}
\hat{R}_k(\tau):=R_k(\vec{s})=\lim\limits_{n\rightarrow\infty}R_{n,k}\left(\vec{t}\right).
\end{align}
We have the following asymptotics
 \begin{equation}\label{eq:LimitD_n}
\lim_{n\to\infty}\frac{D_n(\vec{t})}{D_n(\vec{0})}=\exp\bigg(-\frac{1}{4}\int_0^{\tau} \log\left(\frac{\tau}{\xi}\right) \sum\limits_{k=1}^m\nu_k\hat{R}_k(\xi)d\xi\bigg),
\end{equation}
where %$\nu_k=s_k/\tau=4nt_k/\tau$, and
$\{\hat{R}_k(\tau)\}$ satisfy the following coupled ODEs:
\begin{align*}
&2\tau^2\hat{R}_k''+2(1-\al)\tau \hat{R}_k'+\left(1-\frac{1}{\hat{R}_k}\right)(\tau \hat{R}_k'-\al \hat{R}_k)^2+\hat{R}_k\sum_{j\neq k}\frac{(\tau\hat{R}_j'-\al \hat{R}_j)^2}{\hat{R}_j}\\
&+\frac{\hat{R}_k}{1-\sum\limits_{j=1}^m \hat{R}_j}\cdot\sum_{j=1}^m(\tau \hat{R}_j'-\al \hat{R}_j)\cdot\Bigg(\sum_{j=1}^m(\tau \hat{R}_j'-\hat{R}_j)+2\al \Bigg)+\tau\hat{R}_k(\nu_k-\sum_{j=1}^m \nu_j\hat{R}_j)=0,
\end{align*}
for $k=1,\cdots,m$. Setting
\[\hat{R}_k(\tau)=q_k^2(\tau),\qquad\qquad k=1,\cdots,m,\]
we find that $q_1(\tau),\cdots, q_m(\tau)$ satisfy the differential equations
\begin{equation}\label{eq:EqTW}
\tau q_k (1-\sum_{j=1}^mq_j^2)\sum_{j=1}^m(\tau q_jq_j')'+\tau\Bigg(1-\sum_{j=1}^mq_j^2\Bigg)^2 ((\tau q'_k)'+\frac{a_kq_k}{4})+\tau^2q_k\Bigg(\sum_{j=1}^mq_jq'_j\Bigg)^2-\frac{\alpha^2}{4}q_k=0,
\end{equation}
for $k=1, \cdots, m$.
\end{corollary}

\begin{remark}
Equation \eqref{eq:EqTW} is identical with (1.12) of
\cite{cd} where the asymptotics of $q_1(\tau), \cdots,q_m(\tau)$ as $\tau\to 0$ were provided:
\begin{equation}\label{eq:qasy}q_k(\tau)=\sqrt{c_{k+1}-c_{k}}J_{\alpha}(\sqrt{\nu_k\tau})(1+O(\tau)),\end{equation}
 where $c_{k}=\sum\limits_{j=0}^{k-1} \omega_k/\omega, k=1,\cdots, m$, $c_{m+1}=1$ and $J_{\alpha}$ is the Bessel function of order $\alpha$.
As  pointed out in \cite{cd}, when $m=1$ and $a_1=1$, 
the equation \eqref{eq:EqTW} is reduced to the $P_V$ equation. Corollary \ref{tau-Rq}, together with \eqref{eq:qasy}, generalizes a result of Tracy and Widom  \cite{TW161} for the determinant of the classical Bessel kernel, which corresponds to the case $m = 1$. See also \cite{Charlier,cd} 
for other properties of the Bessel kernel determinants with discontinuities and further applications.
\end{remark}
%equation \eqref{eq:EqTW} is reduced to (1.16) of \cite{TW161}, which is a $P_V$ equation. In this sense, \eqref{eq:EqTW} may be regarded as a system of coupled $P_V$ equations.

%\begin{remark}
%Setting $s_k=a_k\tau$ and $R_k(\vec{s})=\hat{R}_k(\vec{\tau})$ in \eqref{m-RPDE}, and letting $\hat{R}_k(\vec{\tau})=q_j^2(\tau)$ in the resulting equation, we arrive at \eqref{eq:EqTW}.
%\end{remark}

%\end{itemize}

\subsection{Outline of the paper}
The rest of the paper is organized as follows. In Section \ref{sec:LOP}, we employ the ladder operator approach to express the recurrence coefficients in terms of $2m$ auxiliary quantities which satisfy a system of difference equations that can be iterated. We also show that the recurrence coefficients satisfy Toda equations and finally give the proof of Theorem \ref{LO-result}. In Section \ref{sec:RHP}, by introducing the RH problem for the orthogonal polynomials and with the aid of Lax pair, we establish the coupled $P_V$ system presented in Theorem \ref{cpv},  and obtain the relations and expressions given in Theorem \ref{LO-RHP}. %and prove Proposition \ref{cor:Solution}. d
Section \ref{sec:SL} is devoted to the double scaling analysis of the Hankel determinant, and the proof of Proposition \ref{scaledRr}, Theorem \ref{m-resultscaled} and Corollary \ref{tau-Rq}.

\section{Finite Dimension Analysis -- Ladder Operator Approach}\label{sec:LOP}
Write
\begin{align*}
w_0(x):=&x^{\al}\e^{-x}=:\e^{-{\rm v}_0(x)},\\
w_J(x;\vec{t}):=&\omega_0+\sum_{k=1}^m \omega_k\theta(x-t_k).
\end{align*}
Then the weight function \eqref{eq:weight} reads
\[
w(x;\vec{t}\,)=x^{\al}\e^{-x}\left(\omega_0+\sum_{k=1}^m \omega_k\theta(x-t_k)\right)=w_0(x)w_J(x;\vec{t}), \quad x\in[0,\infty).
\]

Using the properties of orthogonal polynomials, as was pointed out in Lemma 1 and Remark 1 of \cite{BasorChen09}, one can show that $\{P_n(x;\vec{t}\,)\}$ satisfy a pair of ladder operators. For simplicity of presentation, in what follows, we shall not display the dependence of $\vec{t}$ when not necessary.

\begin{lemma}\label{lad-oper}
 The monic orthogonal polynomials $P_n(z):=P_n(z;\vec{t}\,)$ satisfy the following lowering and raising operator
\begin{align*}
P_n'(z)=&\beta_nA_n(z)P_{n-1}(z)-B_n(z)P_n(z),\\
P_{n-1}'(z)=&\left(B_n(z)+{\rm v}_0'(z)\right)P_{n-1}(z)-A_{n-1}(z)P_n(z),
\end{align*}
where ${\rm v}_0(z)=z-\al\ln z$, and $A_n(z)$ and $B_n(z)$ are given by
\begin{subequations}\label{defAnBn}
\begin{align}
A_n(z):=&\frac{1}{h_n}\int_0^{\infty}\frac{{\rm v}_0'(z)-{\rm v}_0'(y)}{z-y}P_n^2(y)w(y;\vec{t}\,)dy+\sum_{k=1}^m \frac{R_{n,k}(t_k;\vec{t}\,)}{z-t_k},\label{m-defAn}\\
B_n(z):=&\frac{1}{h_{n-1}}\int_0^{\infty}\frac{{\rm v}_0'(z)-{\rm v}_0'(y)}{z-y}P_n(y)P_{n-1}(y)w(y;\vec{t}\,)dy+\sum_{k=1}^m\frac{r_{n,k}(t_k;\vec{t})}{z-t_k}.\label{m-defBn}
\end{align}
\end{subequations}
Here the auxiliary quantities $R_{n,k}(\vec{t})$ and $r_{n,k}(\vec{t}), k=1,\cdots,m$, are defined by
\begin{align*}
R_{n,k}(\vec{t}):=&\omega_k\frac{w_0(t_k)}{h_n(\vec{t})}P_n^2(t_k;\vec{t}),\\
r_{n,k}(\vec{t}):=&\omega_k\frac{w_0(t_k)}{h_{n-1}(\vec{t})}P_n(t_k;\vec{t})P_{n-1}(t_k;\vec{t}).
\end{align*}
\end{lemma}

\begin{remark} When $-1<\alpha\leq0$, the coefficients  $A_n(z)$ and $B_n(z)$ in the ladder operators should be modified to
\begin{subequations}\label{defMAnBn}
\begin{align}
A_n(z)=&\frac{1}{z}\left(\frac{1}{h_n}\int_0^{\infty}\frac{z{\rm v}_0'(z)-y{\rm v}_0'(y)}{z-y}P_n^2(y)w(y;\vec{t}\,)dy-\sum\limits_{k=1}^m R_{n,k}\right)+\sum_{k=1}^m \frac{R_{n,k}}{z-t_k},\label{m-MdefAn}\\
B_n(z)=&\frac{1}{z}\left(\frac{1}{h_{n-1}}\int_0^{\infty}\frac{z{\rm v}_0'(z)-y{\rm v}_0'(y)}{z-y}P_n(y)P_{n-1}(y)w(y;\vec{t}\,)dy-\sum_{k=1}^mr_{n,k}-n\right)+\sum_{k=1}^m\frac{r_{n,k}}{z-t_k}.\label{MdefBn}
\end{align}
\end{subequations}
They can be derived by using the framework presented in Section 4.3 of \cite{Assche} and \cite{ChenIsmail}. It should be pointed out that \eqref{defMAnBn} are also well-defined for $\al>0$ and in this case they are equivalent to \eqref{defAnBn}.
\end{remark}

With the three-term recurrence relation \eqref{OP-recu} and the above ladder operators, we can derive the following two compatibility conditions for $A_n(z)$ and $B_n(z)$.
\begin{lemma}\label{s1s2}
$A_n(z)$ and $B_n(z)$ defined in \eqref{m-defAn}-\eqref{m-defBn} satisfy the following two identities
\begin{align}
B_{n+1}(z)+B_n(z)=&\left(z-\alpha_n\right)A_n(z)-{\rm v}_0'(z),\tag{$S_1$}\\
1+\left(z-\alpha_n\right)\left(B_{n+1}(z)-B_n(z)\right)=&\beta_{n+1}A_{n+1}(z)-\beta_nA_{n-1}(z).\tag{$S_2$}
\end{align}
\end{lemma}

Multiplying both sides of $(S_2)$ by $A_n(z)$ and eliminating $(z-\al_n)A_n(z)$ in the resulting expression by using $(S_1)$, we find
 \[A_n(z)+B_{n+1}^2(z)-B_n^2(z)+{\rm v}_0'(z)\left(B_{n+1}(z)-B_n(z)\right)=\beta_{n+1}A_{n+1}(z)A_n(z)-\beta_nA_n(z)A_{n-1}(z).\]
Replacing $n$ by $j$ in this equality and summing over $j$ from $0$ to $n-1$, with the initial conditions $B_0(z)=A_{-1}(z)=0$, we obtain another identity for $A_n(z)$ and $B_n(z)$.
\begin{lemma}
We have
\begin{align}
B_n^2(z)+{\rm v}_0'(z)B_n(z)+\sum_{j=0}^{n-1}A_j(z)=&\beta_nA_n(z)A_{n-1}(z).\tag{$S_2'$}
\end{align}
\end{lemma}
Substituting ${\rm v}_0(z)=z-\al\ln z$ into \eqref{defAnBn}, with the aid of \eqref{m-or} and \eqref{defPn}, we obtain the expressions for $A_n(z)$ and $B_n(z)$ in terms of $R_{n,k}(\vec{t})$ and $r_{n,k}(\vec{t})$.
\begin{lemma}\label{AnBnexp}
 $A_n(z)$ and $B_n(z)$ defined in \eqref{m-defAn}-\eqref{MdefBn} are rational functions in $z$ with simple poles at $0,t_1,\cdots,t_m$, reading
\begin{subequations}\label{m-AnBnRr}
\begin{align}
A_n(z)=&\dfrac{1-\sum\limits_{k=1}^m R_{n,k}}{z}+\sum\limits_{k=1}^m\frac{R_{n,k}}{z-t_k},\label{m-AnR}\\
B_n(z)=&-\frac{n+\sum\limits_{k=1}^m r_{n,k}}{z}+\sum_{k=1}^m\frac{r_{n,k}}{z-t_k}.\label{m-Bnr}
\end{align}
\end{subequations}
\end{lemma}
\begin{proof}
We first compute $A_n(z)$ and $B_n(z)$ for the case $\al>0$.
Since
\[\frac{{\rm v}_0'(z)-{\rm v}_0'(y)}{z-y}=\frac{\al}{yz},\]
the integral in \eqref{m-defAn} is given by
\begin{align*}
\frac{1}{h_n}\int_0^{\infty}\frac{{\rm v}_0'(z)-{\rm v}_0'(y)}{z-y}P_n^2(y)w(y;\vec{t})dy
=\frac{1}{z}\cdot\frac{1}{h_n}\int_0^{\infty}P_n^2(y)\e^{-y}w_J(x;\vec{t})dy^{\al}.
\end{align*}
Through integration by parts and noting that $\left(d/dx\right)\theta(x-t)=\delta(x-t)$ with $\delta(\cdot)$ denoting the Dirac delta function, we find
\begin{align}
&\frac{1}{h_n}\int_0^{\infty}\frac{{\rm v}_0'(z)-{\rm v}_0'(y)}{z-y}P_n^2(y)w(y;\vec{t})dy\nonumber\\
=&\frac{1}{z}\cdot\frac{1}{h_n}\left[P_n^2(y)w(y;\vec{t})\Big|_{y=0}^{y=\infty}-\int_0^{\infty}2P_n(y)P_n'(y)w(y;\vec{t})dy\right.\nonumber\\
&\left.-\int_0^{\infty}P_n^2(y)y^{\al}(-\e^{-y})w_J(x;\vec{t})dy-\sum_{k=1}^m w_k\int_0^{\infty}P_n^2(y)y^{\al}\e^{-y} \delta(y-t_k)dy\right].\label{m-An1}
\end{align}
Now we look at the four terms in the square bracket. Noting that $w(0;\vec{t})=w(\infty;\vec{t})=0$, the first term is equal to zero. Since
\begin{align*}
P_n'(y)=ny^{n-1}+\cdots=nP_{n-1}(y)+\sum_{k=0}^{n-2}d_kP_{k}(y),
\end{align*}
according to the orthogonality \eqref{m-or}, we find that the second term is also zero and the third one is $h_n$. Due to the following property of the Dirac delta function
\[\int_{0}^{\infty}f(x)\delta(x-t)dx=f(t),\]
the last term is seen to be $h_n\sum\limits_{k=1}^m R_{n,k}$.
Hence we conclude that $\eqref{m-An1}$ equals $(1-\sum\limits_{k=1}^m R_{n,k})/z$. Plugging it back into \eqref{m-defAn}, we get \eqref{m-AnR}.
The integral in \eqref{m-defBn} can be computed similarly and thus we omit the derivation of \eqref{m-Bnr}.

For $-1<\al\leq0$, we make use of \eqref{defMAnBn} to calculate $A_n(z)$ and $B_n(z)$, which gives us \eqref{m-AnBnRr} immediately.
\end{proof}

Inserting \eqref{m-AnBnRr} into the three identities satisfied by $A_n(z)$ and $B_n(z)$, namely, $(S_1), (S_2), (S_2')$,  by equating the residues on their both sides, we arrive at a series of difference equations.

\subsection{Difference equations for auxiliary quantities}
From $(S_1)$, we get
\begin{align}
z^{-1}:\qquad\qquad2n+1+\sum_{k=1}^m\left(r_{n+1,k}+r_{n,k}\right)=&\al_n\left(1-\sum_{k=1}^mR_{n,k}\right)-\al,\label{m-S11}\\
(z-t_k)^{-1}:\qquad\qquad\qquad\qquad\qquad r_{n+1,k}+r_{n,k}=&(t_k-\al_n)R_{n,k}, \qquad k=1,\cdots,m.\qquad\label{m-S12}
\end{align}
From $(S_2)$, we find
\begin{align}
\quad z^{-1}:\qquad \al_n\left(1+\sum_{k=1}^m\left(r_{n+1,k}-r_{n,k}\right)\right)=&\bt_{n+1}\left(1-\sum_{k=1}^mR_{n+1,k}\right)-\beta_n\left(1-\sum_{k=1}^mR_{n-1,k}\right),\label{S21}\\
(z-t_k)^{-1}:\qquad\left(t_k-\al_n\right)\left(r_{n+1,k}-r_{n,k}\right)=&\beta_{n+1}R_{n+1,k}-\beta_nR_{n-1,k},\qquad k=1,\cdots,m.\label{S22}
\end{align}
From $(S_2')$, we obtain
%\begin{equation*}
%\begin{aligned}
%z^{-1}:\qquad\qquad&-\sum_{k=1}^mr_{n,k}+\left(2n+\al+2\sum_{k=1}^m r_{n,k}\right)\cdot\sum_{k=1}^m\frac{r_{n,k}}{t_k}-\sum_{j=0}^{n-1}\sum_{k=1}^mR_{j,k}\\
%&=\bt_n\sum_{k=1}^m\frac{1}{t_k}\left[R_{n-1,k}\Bigg(\sum_{j=1}^mR_{n,j}-1\Bigg)+R_{n,k}\Bigg(\sum_{j=1}^mR_{n-1,j}-1\Bigg)\right],\qquad\qquad
%\end{aligned}
%\end{equation*}
\begin{gather}\label{m-S2'2}
\quad z^{-2}:\quad\qquad\Bigg(n+\al+\sum_{k=1}^m r_{n,k}\Bigg)\left(n+\sum_{k=1}^m r_{n,k}\right)=\bt_n\left(1-\sum_{k=1}^m R_{n,k}\right)\left(1-\sum_{k=1}^m R_{n-1,k}\right),
\end{gather}
%\begin{equation*}
%\begin{aligned}
%(z-t_k)^{-1}:\;\qquad\quad&-\frac{r_{n,k}}{t_k}\left(2n+\al+2\sum_{j=1}^m r_{n,j}\right)+r_{n,k}+\sum_{j=1,j\neq k}^m\frac{2r_{n,j}r_{n,k}}{t_k-t_j}+\sum_{j=0}^{n-1}R_{j,k}\\
%&=\bt_n\left[\frac{R_{n-1,k}(1-\sum_{j=1}^mR_{n,j})+R_{n,k}\left(1-\sum_{j=1}^mR_{n-1,j}\right)}{t_k}\right.\\
%&\qquad\qquad\left.+\sum_{j\neq k}\frac{R_{n,k}R_{n-1,j}+R_{n-1,k}R_{n,j}}{t_k-t_j}\right],\qquad k=1,\cdots,m,\qquad
%\end{aligned}
%\end{equation*}
\begin{gather}\label{m-S2'1}
(z-t_k)^{-2}:\qquad\qquad\qquad\qquad r_{n,k}^2=\bt_nR_{n,k}R_{n-1,k},\qquad k=1,\cdots,m.\qquad\qquad\qquad\quad
\end{gather}

Using \eqref{m-S11}, \eqref{m-S12}, \eqref{m-S2'2} and \eqref{m-S2'1}, we can express $\al_n(\vec{t})$ and $\bt_n(\vec{t})$ in terms of $\{R_{n,k}(\vec{t})\}$ and $\{r_{n,k}(\vec{t})\}$ which are shown to satisfy a system of difference equations.
\begin{lemma}
The recurrence coefficients are expressed in terms of the auxiliary quantities by
\begin{align}
\al_n(\vec{t})=&\sum_{k=1}^m t_k R_{n,k}+2n+1+\al,\label{m-alR}\\
\bt_n(\vec{t})=&\frac{\left(n+\al+\sum\limits_{k=1}^m r_{n,k}\right)\left(n+\sum\limits_{k=1}^m r_{n,k}\right)}{1-\sum\limits_{k=1}^m R_{n,k}}+\sum\limits_{k=1}^m\frac{r_{n,k}^2}{R_{n,k}}.\label{btRr}
\end{align}
\end{lemma}
\begin{proof}
Summing \eqref{m-S12} over $k=1,\cdots,m$, and combining the resulting equality with \eqref{m-S11}, we are led to \eqref{m-alR}.
Eliminating $R_{n-1,k}$ in \eqref{m-S2'2} by using \eqref{m-S2'1}, after simplification, we obtain \eqref{btRr}.
\end{proof}

\begin{lemma}
The quantities $\{R_{n,k}(\vec{t})\}$ and $\{r_{n,k}(\vec{t})\}$ satisfy the following system of difference equations
\begin{align}
r_{n+1,k}=&\Bigg(t_k-\sum_{j=1}^m t_j R_{n,j}-2n-1-\al\Bigg)R_{n,k}-r_{n,k},\qquad k=1,\cdots,m,\label{diff12}\\
\frac{1}{R_{n,1}}=&1+\frac{R_{n-1,1}}{r_{n,1}^2}\cdot\left[\sum_{k=2}^m\frac{r_{n,k}^2}{R_{n-1,k}}
+\frac{\left(n+\al+\sum\limits_{k=1}^m r_{n,k}\right)\left(n+\sum\limits_{k=1}^m r_{n,k}\right)}{1-\sum\limits_{k=1}^m R_{n-1,k}}\right],\label{diff3}\\
R_{n,k}=&\frac{R_{n-1,1}}{R_{n-1,k}}\cdot\frac{r_{n,k}^2}{r_{n,1}^2}\cdot R_{n,1} ,\qquad k=2,\cdots,m,\label{diff4}
\end{align}
 which can be iterated in $n$ with initial conditions
\begin{align*}
R_{0,k}=\frac{\omega_k t_k^{\al}\e^{-t_k}}{\int_0^\infty w(x;\vec{t})dx}, \qquad\qquad r_{0,k}=0.
\end{align*}
\end{lemma}
\begin{proof}
Plugging \eqref{m-alR} into \eqref{m-S12} leads to \eqref{diff12} immediately. With $k=1$ in \eqref{m-S2'1}, we have
\begin{align}\label{btrn1}
\bt_n=\frac{r_{n,1}^2}{R_{n,1}R_{n-1,1}}.
\end{align}
Substituting it into \eqref{m-S2'1} with $k=2,\cdots,m$, we come to \eqref{diff4}.
Combining \eqref{btrn1} with \eqref{btRr}, we get \eqref{diff3}.
\end{proof}

To conclude this subsection, we present expressions for $p(n,\vec{t})$, the second leading coefficient of degree-$n$ monic orthogonal polynomial $P_n(x;\vec{t})$, involving $\bt_n$ and the auxiliary quantities.
We will see that these relations play an important role in the study of the log derivative of the Hankel determinant.

\begin{lemma} We have
\begin{align}
p(n,\vec{t})=&\sum_{k=1}^m t_k r_{n,k}-\bt_n\label{m-btp}\\
=&-\sum_{k=1}^mt_k\sum_{j=0}^{n-1}R_{j,k}-n(n+\al).\label{m-psumR}
\end{align}
\end{lemma}
\begin{proof}
From \eqref{m-alp}, it follows that
\begin{align}\label{m-sumalp}
\sum_{j=0}^{n-1}\alpha_j(\vec{t})=-p(n,\vec{t}).
\end{align}
This motivates us to derive \eqref{m-btp} and \eqref{m-psumR} using the expressions involving $\al_n$.

Summing \eqref{S22} over $k=1,\cdots,m$, and adding the resulting equality to \eqref{S21}, we get
\begin{align}\label{aldiff}
\al_n+\sum_{k=1}^m t_k(r_{n+1,k}-r_{n,k})=\beta_{n+1}-\beta_n.
\end{align}
Replacing $n$ by $j$ and summing it from $j=0$ to $n-1$, noting that $r_{0,k}=0=\bt_0$ and in view of \eqref{m-sumalp}, we
obtain \eqref{m-btp}.
Replacing $n$ by $j$ in \eqref{m-alR} and summing it over $j$ from $0$ to $n-1$, we are led to \eqref{m-psumR}.
\end{proof}

%\begin{remark}
%\eqref{psumR} can also be derived by using \eqref{S2'3} and \eqref{S2'4}, as is done in \cite{BasorChenZhang} for the special case where $A=0, B_1=1, B_2=-1$. However, it seems to the authors that our proof is more straightforward.
%\end{remark}

\subsection{Toda equations for recurrence coefficients}
%In the previous subsection, we have established several difference equations and relationships which connect the recurrence coefficients to the auxiliary quantities.
We proceed to develop differential relations by differentiating the orthogonality relation \eqref{m-or} with $m=n$ and $m=n-1$.
It turns out that the derivatives of $\ln h_n(\vec{t})$ and $p(n,\vec{t})$ are directly related to the auxiliary quantities.

\begin{lemma} We have
\begin{align}
\partial_{t_k}\ln h_n(\vec{t})=&-R_{n,k}(\vec{t}),\label{DhR}\\
\partial_{t_k}p(n,\vec{t})=&r_{n,k}(\vec{t}),\label{Dpr}
\end{align}
which, according to $\bt_n=h_n/h_{n-1}$ and $\al_n=p(n,\vec{t})-p(n+1,\vec{t})$, gives us
\begin{align}
\partial_{t_k}\ln \beta_n(\vec{t})=&R_{n-1,k}(\vec{t})-R_{n,k}(\vec{t}),\label{DbtR}\\
\partial_{t_k}\al_n(\vec{t})=&r_{n,k}(\vec{t})-r_{n+1,k}(\vec{t}).\label{Dalr}
\end{align}
\end{lemma}
\begin{proof}
Differentiating both sides of
\[h_n(\vec{t})=\int_0^{\infty}P_n^2(x;\vec{t})w(x;\vec{t})dx\]
over $t_k$, we get
\begin{align}\label{m-Dhn1}
\partial_{t_k}h_n(\vec{t})=\int_0^{\infty}2P_n(x;\vec{t})\cdot\left(\partial_{t_k}P_n(x;\vec{t})\right)w(x;\vec{t}))dx+\int_0^{\infty}P_n^2(x;\vec{t})\cdot\left(\partial_{t_k}w(x;\vec{t})\right)dx.
\end{align}
Applying $\partial_{t_k}$ to
\[P_n\left(x;\vec{t}\,\right)=x^n+p(n,\vec{t}\,)x^{n-1}+\cdots+P_n(0;\vec{t}\,),\]
we have
\begin{align*}
\partial_{t_k}P_n(x;\vec{t})=&\partial_{t_k}p(n,\vec{t}\,)\cdot x^{n-1}+\{\text{lower~degree~of~}x\}\\
=&\partial_{t_k}p(n,\vec{t}\,)\cdot P_{n-1}+\{\text{linear~combination~of~}P_{i}, i=n-2,n-3,\cdots,1,0\}.
\end{align*}
This together with \eqref{m-or} indicates that the first integral in \eqref{m-Dhn1} is zero and hence
\begin{align*}
\partial_{t_k}h_n(\vec{t})=&\int_0^{\infty}P_n^2(x;\vec{t})\cdot\left(\partial_{t_k}w(x;\vec{t})\right)dx\\
=&- w_k\int_0^{\infty}P_n^2(x;\vec{t})x^{\al}\e^{-x}\delta(x-t_k)dx\\
=&-R_{n,k}(\vec{t})\cdot h_n(\vec{t}),
\end{align*}
which gives us \eqref{DhR}.

Taking the derivative of \[0=\int_0^{\infty} P_n(x;\vec{t})P_{n-1}(x;\vec{t})dx\]
with respect to $t_k$, via an argument similar to the above, we can derive \eqref{Dpr} .
\end{proof}

By using the differential relations \eqref{DbtR} and \eqref{Dalr}, and with the aid of identities and expressions involving $\al_n$ and $\bt_n$ which are presented in the previous subsection, we establish Toda equations for $\al_n$ and $\bt_n$.
\begin{proposition} The recurrence coefficients $\alpha_n$ and $\beta_n$ satisfy the following Toda equations
\begin{subequations}\label{Toda}
\begin{equation}\label{Todabt}
\delta\ln\beta_n=\al_{n-1}-\al_n+2,
\end{equation}
\begin{equation}\label{Todal}
\left(\delta-1\right)\al_n=\beta_n-\beta_{n+1},
\end{equation}
where $\delta=\sum\limits_{k=1}^m t_k\partial_{t_k}$.
\end{subequations}
\end{proposition}
\begin{proof}
Multiplying both sides of \eqref{DbtR} by $t_k$ yields
\[
t_k\,\partial_{t_k}\ln \beta_n=t_k\left(R_{n-1,k}-R_{n,k}\right).\]
Summing this equation over $k$ from $1$ to $m$, we have
\begin{align*}
\delta\ln\beta_n=&\sum_{k=1}^m t_k R_{n-1,k}-\sum_{k=1}^mt_kR_{n,k}.
\end{align*}
Using \eqref{m-alR} to remove the two summation terms, we get \eqref{Todabt}.
Due to \eqref{Dalr}, we replace $r_{n+1,k}-r_{n,k}$ by $-\partial_{t_k}\al_n$ in \eqref{aldiff} and obtain \eqref{Todal}.
\end{proof}

\subsection{Proof of Theorem \ref{LO-result}}

{\bf (a) Derivation of the Riccati equations \eqref{m-Ric} and the coupled PDEs \eqref{m-RPDE}}\\
%Now we go ahead with the derivation of Riccati-type equations for $R_{n,k}$ and $r_{n,k}$. To do this,
We observe the following facts:
\begin{align}
\partial_{t_k} R_{n,j}=&\partial_{t_j} R_{n,k},\label{m-DklR}\\
\partial_{t_k} r_{n,j}=&\partial_{t_j} r_{n,k},\label{m-Dklr}
\end{align}
for $j,k=1,\cdots,m$. The first equality can be seen by combining the fact that $\partial_{t_{j}t_k}^2\ln h_n=\partial_{t_k t_{j}}^2\ln h_n$ with \eqref{DhR}, while the second one is a consequence of $\partial_{t_{j}t_k}^2p(n,\vec{t})=\partial_{t_k t_{j}}^2p(n,\vec{t})$ and \eqref{Dpr}.

Now we go ahead to derive the Riccati equations satisfied by $R_{n,k}$ and $r_{n,k}$.
Eliminating $r_{n+1,k}$ from \eqref{m-S12} and \eqref{Dalr}, we get
\begin{align*}
\partial_{t_k}\al_n=2r_{n,k}+(\al_n-t_k)R_{n,k}.
\end{align*}
Inserting \eqref{m-alR} into this equation, we find
\begin{align*}
\sum_{j\neq k} t_j\cdot\partial_{t_k} R_{n,j}+t_k\cdot\partial_{t_k}R_{n,k}=2r_{n,k}+\Bigg(\sum_{j=1}^m t_j R_{n,j}-t_k+2n+\al\Bigg)R_{n,k}.
\end{align*}
According to \eqref{m-DklR}, we replace $\partial_{t_k} R_{n,j}$ by $\partial_{t_j} R_{n,k}$ in the above equation, which leads to the first Riccati equation \eqref{m-Ric1}.
To derive \eqref{m-Ric2}, we make use of \eqref{m-S2'1} and \eqref{DbtR}, both of which are related to $R_{n-1,k}$ and $R_{n,k}$:
\begin{align*}
r_{n,k}^2=&\bt_nR_{n,k}R_{n-1,k},\\
\partial_{t_k}\bt_n=&\bt_nR_{n-1,k}-\bt_nR_{n,k}.
\end{align*}
Getting rid of $R_{n-1,k}$ from the second equality by using the first one, we get
\begin{align}\label{Dbtn}
\partial_{t_k}\bt_n=\frac{r_{n,k}^2}{R_{n,k}}-\bt_nR_{n,k}.
\end{align}
Applying $\partial_{t_k}$ to \eqref{m-btp}, in view of \eqref{Dpr} and \eqref{m-Dklr}, we find
\begin{align*}
\partial_{t_k}\bt_n=\delta r_{n,k}.
\end{align*}
Combining it with \eqref{Dbtn}, and replacing $\bt_n$ in the resulting equation by using \eqref{btRr}, we come to \eqref{m-Ric2}, which completes the derivation of Riccati equations.

Solving $r_{n,k}$ from \eqref{m-Ric1} and substituting it into \eqref{m-Ric2}, noting that
$\delta(t_k)=t_k$ and $\delta(t_jR_{n,j})=t_jR_{n,j}+t_j\cdot\delta R_{n,j}$, after simplification, we arrive at
the coupled PDEs satisfied by $R_{n,k}, k=1,2,\cdots, m$.
\medskip\\
{\bf (b) Derivation of PDE \eqref{m-sigmaPDE} satisfied by $\sigma_n(\vec{t})$}\\
Recall the Hankel determinant of our interest, i.e.
\begin{align*}
D_n\left(\,\vec{t}\,\right):=\det\left(\int_{0}^{\infty} x^{i+j}w(x;\vec{t}\,)dx\right)_{i,j=0}^{n-1}.
\end{align*}
It is well known that it can be represented as the product of $h_j$, the square of the $L^2$ norm of the monic orthogonal polynomial $P_j(x;\vec{t})$, that is,
\begin{align*}
D_n(\vec{t})=\prod_{j=0}^{n-1}h_j(\vec{t}).
\end{align*}
Refer to \cite[pp. 16--19]{Ismail}.
Thus, it follows from \eqref{DhR} that
\begin{align*}
\sigma_n(\vec{t})=\delta\,ln D_n(\vec{t})=&-\sum_{j=0}^{n-1}\sum_{k=1}^{m}t_kR_{j,k}(\vec{t}),
\end{align*}
where $\delta=\sum\limits_{k=1}^m t_k\partial_{t_k}$.
According to \eqref{m-btp} and \eqref{m-psumR}, we find
\begin{align}
\sigma_n=&p(n,\vec{t})+n(n+\al)\label{sigmap}\\
=&-\beta_n+\sum_{k=1}^m t_k r_{n,k}+n(n+\al).\label{sigmabtr}
\end{align}
Inserting \eqref{btRr} into \eqref{sigmabtr}, we obtain the expression for $\sigma_n(\vec{t})$ in terms of $R_{n,k}$ and $r_{n,k}$, namely \eqref{sigmarR}. If, in turn, we can express $R_{n,k}$ and $r_{n,k}$ by $\sigma_n$ or its derivatives, then we will readily establish the PDE satisfied by $\sigma_n$.
\begin{lemma}\label{rnRnsigma}
$r_{n,k}$ and $R_{n,k}$ are connected with $\sigma_n$ and its derivatives by
\begin{align}\label{rDsigma}
r_{n,k}=\partial_{t_k}\sigma_n,
\end{align}
and
\begin{subequations}\label{Rbtsigma}
\begin{align}\label{Rsigma}
R_{n,k}=\frac{1}{2\bt_n}\left(-\partial_{t_k}\bt_n+{\rm sgn}(\omega_k)\sqrt{(\partial_{t_k}\bt_n)^2+4\bt_nr_{n,k}^2}\right),
\end{align}
for $k=1,\cdots,m$, where ${\rm sgn}(\omega_k)$ is the sign function of $\omega_k$, which is $1$ for $\omega_k>0$, $-1$ for $\omega_k<0$ and $0$ for $\omega_k=0$.
Here
\begin{align}
\bt_n=&\delta\sigma_n-\sigma_n+n(n+\al),\label{btsigma}\\
\partial_{t_k}\bt_n=&\sum_{j=1}^m t_j\cdot\partial_{t_kt_j}^2\sigma_n.\label{D1btsigma}
\end{align}
\end{subequations}
\end{lemma}
\begin{proof}
Differentiating both sides of \eqref{sigmap} with respect to $t_k$, in light of \eqref{Dpr}, we get \eqref{rDsigma}. Plugging it into \eqref{sigmabtr} gives us \eqref{btsigma} which immediately leads to \eqref{D1btsigma}.

Note that \eqref{Dbtn} is equivalent to the following second order linear equation in $R_{n,k}$:
\begin{align*}
\bt_nR_{n,k}^2+\left(\partial_{t_k}\bt_n\right)R_{n,k}-r_{n,k}^2=0.
\end{align*}
Solving for $R_{n,k}$ from it, we get two possible solutions
\begin{align*}
R_{n,k}=\frac{-\partial_{t_k}\bt_n\pm\sqrt{(\partial_{t_k}\bt_n)^2+4\bt_nr_{n,k}^2}}{2\bt_n},
\end{align*}
or equivalently,
\begin{align}\label{Rsol}
R_{n,k}+\frac{1}{2}\partial_{t_k}\ln\bt_n=\frac{\pm\sqrt{(\partial_{t_k}\bt_n)^2+4\bt_nr_{n,k}^2}}{2\bt_n}.
\end{align}
To determine which solution to choose, it suffices to discuss the sign function of $R_{n,k}+\frac{1}{2}\partial_{t_k}\ln\bt_n$. According to \eqref{DbtR}, we have
\begin{align}\label{Rsol2}
R_{n,k}+\frac{1}{2}\partial_{t_k}\ln\bt_n=\frac{1}{2}\left(R_{n-1,k}+R_{n,k}\right).
\end{align}
From the definition of $R_{j,k}$, i.e.
\begin{align*}
R_{j,k}=\frac{\omega_k}{h_j}P_j^2(t_k;\vec{t})t_k^{\al}\e^{-t_k}, \qquad\qquad j\geq0,
\end{align*}
and the fact that $h_j>0$, we observe that for $t_k>0$,
\begin{align*}
{\rm sgn}(R_{n-1,k})={\rm sgn}(R_{n,k})={\rm sgn}(\omega_k).
\end{align*}
Hence, it follows from \eqref{Rsol2} that
\[{\rm sgn}\left(R_{n,k}+\frac{1}{2}\partial_{t_k}\ln\bt_n\right)={\rm sgn}(\omega_k).\]
This combined with \eqref{Rsol} gives us \eqref{Rsigma}.
\end{proof}

Substituting \eqref{rDsigma} and \eqref{Rbtsigma} into \eqref{sigmarR}, after simplification, we arrive at \eqref{m-sigmaPDE}.

\section{Lax Pair and Coupled $P_V$ System} \label{sec:RHP}
We start with the RH problem for the monic polynomials $P_n(z;\vec{t})$ orthogonal with respect to the weight function \eqref{eq:weight}.
%\[
%w(x;\vec{t}\,)=x^{\al}\e^{-x}\left(\omega_0+\sum_{k=1}^m \omega_k\theta(x-t_k)\right).
%\]
With the aid of Lax pair, we can express $\sigma_n(\vec{t})$ in terms of quantities satisfying a coupled $P_V$ system.

\subsection{RH problem for orthogonal polynomials }\label{sec:OP-RHP}
Let
\begin{equation}\label{def:Y}
Y(z;\vec{t})= \left (\begin{array}{cc}
P_n(z)& \frac 1{2\pi i}\int_0^{+\infty} \frac{P_n(x) w(x;\vec{t})}{x-z}dx\\
-2\pi i \gamma_{n-1}^2 \;P_{n-1}(z)& - \gamma_{n-1}^2\;
\int_0^{+\infty} \frac{P_{n-1}(x) w(x;\vec{t})}{x-z}dx\end{array} \right ),
\end{equation}
where $\gamma_{n-1}=h_{n-1}^{-1/2}$ with $h_{n-1}$ denoting the square of the $L^2$-norm of $P_{n-1}(z)$ with respect to $w(x;\vec{t})$ defined in \eqref{eq:weight}.

From \eqref{def:Y}, it is  seen that the function $Y(z)=Y(z;\vec{t})$ is the unique solution of the following RH problem (see Fokas, Its and Kitaev \cite{FIK}, and also Deift \cite{Deift}):

\begin{description}
  \item(a)~~  $Y(z)$ is analytic in
  $\mathbb{C}\backslash [0,+\infty)$;

  \item(b)~~  $Y(z)$  satisfies the jump condition
  $$Y_+(x)=Y_-(x) \left(
                               \begin{array}{cc}
                                 1 & w(x;\vec{t}) \\
                                 0 & 1 \\
                                 \end{array}
                             \right),
\qquad x\in (0,+\infty),$$
where $w(x;\vec{t})$ is defined in \eqref{eq:weight};

  \item(c)~~  As $z\to\infty$, we have
  \begin{equation*}\label{eq:Y-infinity}Y(z)=\left (I+\frac {Y_{1}}{z}+\frac {Y_{2}}{z^2}+O\left (\frac 1 {z^3}\right )\right )\left(
                               \begin{array}{cc}
                                 z^n & 0 \\
                                 0 & z^{-n} \\
                               \end{array}
                             \right),\quad \quad z\rightarrow
                             \infty ;
                             \end{equation*}
\item(d)~~  As $z\to 0$, we have

\begin{equation*}
 Y(z)=\left\{
  \begin{array}{ccc}
                                 \left(
                               \begin{array}{cc}
                                 O(1)& O(1) \\
                                 O(1) & O(1) \\
                                 \end{array}
                             \right), & \quad \alpha>0,\\
                                \left(
                               \begin{array}{cc}
                                 O(1)& O(\log z)\\
                                 O(1) &O(\log z) \\
                                 \end{array}
                             \right), & \quad \alpha=0,\\

                          \left(
                               \begin{array}{cc}
                                 O(1)& O(z^{\alpha}) \\
                                 O(1) & O(z^{\alpha}) \\
                                 \end{array}
                             \right), &\quad -1< \alpha<0;
                               \end{array}
 \right.
 \end{equation*}

%  \begin{equation}\label{eq: Y-origin-1}
% Y(z)=Y^{(0)}(z)\left(
%                               \begin{array}{cc}
%                                 1 & \frac{b_0}{1-e^{2\pi i\alpha}}z^{\alpha} \\
%                                 0 & 1 \\
%                                 \end{array}
%                             \right) , \quad \alpha \notin \mathbb{N}, \end{equation}
%  and
%    \begin{equation}\label{eq: Y-origin-2}
% Y(z)=Y^{(0)}(z)\left(
%                               \begin{array}{cc}
%                                 1 & -\frac{b_0}{2\pi i}z^{\alpha} \log z  \\
%                                 0 & 1 \\
%                                 \end{array}
%                             \right) , \quad \alpha\in \mathbb{N}, \end{equation}
%

%                             where  $Y^{(0)}(z)$ is analytic near the origin,  the branch cut for $z^{\alpha}$ and $\log z$ is taken along $(0, +\infty)$ such  that $\arg z\in(0,2\pi)$.

\item(e)~~  As $z\to t_k$,  we have
\begin{equation*}
 Y(z)=\left(
                               \begin{array}{cc}
                                 O(1) & O( \log (z-t_k) ) \\
                                 O(1) & O( \log (z-t_k) )  \\
                                 \end{array}
                             \right),\end{equation*}
for $k=1,\cdots, m$.
%\begin{equation}\label{eq: Y-sk}
% Y(z)=\left(
%                               \begin{array}{cc}
%                                 1 & -\frac{b_k}{2\pi i} \log (z-s_k)  \\
%                                 0 & 1 \\
%                                 \end{array}
%                             \right) ,  \end{equation}

%                             where  $Y^{(k)}(z)$ is analytic near $s_k$,  the branch cut for $\log (z-s_k)$ is taken along $(s_k, +\infty)$ such  that $\arg (z-s_k)\in(0,2\pi)$, for $k=1,2$, respectively.

\end{description}

\subsection{A model RH problem and Lax pair }\label{sec:model-RHP}

We define
\begin{equation}\label{def:Psi}
 \Psi(z; \vec{t})=t_1^{-(n+\frac{\alpha}{2}) \sigma_3}Y(t_1z;\vec{t})(t_1z)^{\frac{\alpha}{2} \sigma_3} e^{-\frac{t_1}{2}z\sigma_3}, \end{equation}
where $\arg z\in (-\pi,\pi)$ and $\sigma_3$ is the Pauli matrix
\[\sigma_3=\left(
                               \begin{array}{cc}
                              1&  0\\
                                 0 &  -1\\
                                 \end{array}
                             \right).\]
Then, it follows from the RH problem for $Y$ that  $\Psi(z)=\Psi(z;\vec{t})$ satisfies the following RH problem.

\begin{description}
  \item(a)  $\Psi(z)$ is analytic   in
  $\mathbb{C}\backslash [0,+\infty)$;

  \item(b) $\Psi(z)$  satisfies the jump condition
  \begin{equation}\label{eq:Psi-Jump}
\left\{  \begin{array}{ll}
\Psi_+(x)=\Psi_-(x)
                               \begin{pmatrix}
                              1&  \omega_0+\sum_{k=1}^m\omega_k\theta(x-a_k)\\
                                 0 &  1\\
                                 \end{pmatrix},
& \quad x\in (0,+\infty),\medskip\\
\Psi_+(x)=\Psi_-(x) e^{\pi i\alpha\sigma_3},& \quad x\in(-\infty,0),
\end{array}\right.
\end{equation}
where
\[a_k:=\frac{t_k}{t_1},\qquad k=1,\cdots, m;\]

  \item(c)~~  As $z\to\infty$, we have
  \begin{equation}\label{eq:Psi-infty} \Psi(z)=\left (I+\frac {\Psi_{1}}{z}+\frac {\Psi_{2}}{z^2}+O\left (\frac 1 {z^3}\right )\right )                                 z^{(n+\frac{\alpha}{2})\sigma_3} {\e}^{-\frac{t_1}{2}z\sigma_3};\end{equation}
\item(d)~~  As $z\to 0$, we have

  \begin{equation}\label{eq: Psi-origin-1}
 \Psi(z)=\Psi^{(0)}(z)z^{\frac{\alpha}{2} \sigma_3} \left(
                               \begin{array}{cc}
                                 1 & \frac{\omega_0}{1-{\e}^{2\pi i\alpha}}  \\
                                 0 & 1 \\
                                 \end{array}
                             \right) E_0, \quad \alpha \notin \mathbb{N}, \end{equation}
  and
    \begin{equation}\label{eq: Psi-origin-2}
 \Psi(z)=\Psi^{(0)}(z)z^{\frac{\alpha}{2} \sigma_3} \left(
                               \begin{array}{cc}
                                 1 & -\frac{\omega_0}{2\pi i} \log z  \\
                                 0 & 1 \\
                                 \end{array}
                             \right) E_0, \quad \alpha\in \mathbb{N}, \end{equation}

                             where  $\Psi^{(0)}(z)$ is analytic near the origin,  the branch cut for $z^{\alpha}$ and $\log z$ is taken along $(0, +\infty)$ such  that $\arg z\in(-\pi,\pi)$ and  $E_0$ is a piecewise constant matrix such that
                             $E_0=I$ for $\arg z\in(0,\pi)$ and $E_0=\left(
                               \begin{array}{cc}
                                 1 & -\omega_0 \\
                                 0 & 1 \\
                                 \end{array}
                             \right)$ for $\arg z\in(-\pi,0)$;

\item(e)~~  As $z\to a_k$, we have

\begin{equation}\label{eq: Psi-ak}
 \Psi(z)=\Psi^{(k)}(z)\left(
                               \begin{array}{cc}
                                 1 & -\frac{\omega_k}{2\pi i} \log (z-a_k)  \\
                                 0 & 1 \\
                                 \end{array}
                             \right) E_k,  \end{equation}

                             where $\Psi^{(k)}(z)$ is analytic near $a_k$,  the branch cut for $\log (z-a_k)$ is taken along $(a_k, +\infty)$ such  that $\arg (z-a_k)\in(-\pi,\pi)$, for  $k=1,\cdots, m$.  The piecewise constant matrix $E_k$ is defined by                              $E_k=I$ for $\arg z\in(0,\pi)$ and $E_k=\left(
                               \begin{array}{cc}
                                 1 & -\omega_k \\
                                 0 & 1 \\
                                 \end{array}
                             \right)$ for $\arg z\in(-\pi,0)$, for $k=1,\cdots, m$.

\end{description}

\begin{proposition} We have the following  Lax pair
\begin{equation}\label{eq: Lax pair}
\left\{
\begin{array}{l}
  \frac{d}{dz} \Psi(z;\vec{t})=L(z;\vec{t}) \Psi(z;\vec{t}),   \\
  \delta \Psi(z;\vec{t})=U(z;\vec{t}) \Psi(z;\vec{t}) ,\\
  \partial_{t_k}\Psi(z;\vec{t})=\frac{\hat{A}_k(t)}{t_k-t_1z} \Psi(z;\vec{t}), \quad k=2,\cdots, m,
%\Psi_{a_k}(z;t)=\frac{A_k(t)}{a_k-z} \Psi(z;t), \quad k=1,\cdots, m,
\end{array}
\right.
  \end{equation}
  where $\delta=\sum\limits_{k=1}^m t_k\partial_{t_k}$ and
  \begin{equation}\label{eq: L}
 L(z;\vec{t})=-\frac{t_1}{2}\sigma_3+\frac{\hat{A}_0(\vec{t})}{z}+\sum_{k=1}^m \frac{\hat{A}_k(\vec{t})}{z-a_k},   \end{equation}
  \begin{equation}\label{eq: U}
 U(z;\vec{t})=-\frac{t_1}{2}z\sigma_3+\hat{B}(\vec{t}),
\end{equation}
with $a_k=t_k/t_1, k=1,\cdots, m$. Here the coefficients  are given below
   \begin{equation}\label{eq:A0}
 \hat{A}_0= \begin{pmatrix}
                    a &  -\left( a-\frac{\alpha}{2}\right)y\\
                 \left(a+\frac{\alpha}{2}\right)/y &-a
                  \end{pmatrix},  \end{equation}
  \begin{equation}\label{eq:Ak}
 \qquad\qquad \hat{A}_k= \begin{pmatrix}
                   - u_kv_k &  u_ky\\
                -u_kv_k^2/y &  u_kv_k               \end{pmatrix},  \quad k=1,\cdots, m, \end{equation}
                and
  \begin{equation}\label{eq:B}
\hat{B}= \begin{pmatrix}
                  0 &  b_1y\\
           b_2/y & 0                  \end{pmatrix},   \end{equation}
with $b_1(\vec{t})$ and $b_2(\vec{t})$ defined by \eqref{eq:b-k}, and
 \begin{equation*}
 a(\vec{t}):= \sum\limits_{k=1}^mu_k(\vec{t})v_k(\vec{t})+n+\frac{\alpha}{2}.
\end{equation*}
 \end{proposition}

It is seen from \eqref{eq:Psi-Jump} that $\Psi$, $\frac{d}{dz}\Psi$, $\partial_{t_k}\Psi$ and $\delta\Psi$ share the same jump \eqref{eq:Psi-Jump} on the real axis.
 Thus, $L=(\frac{d}{dz}\Psi)\Psi^{-1}$, $U=(\delta\Psi)\Psi^{-1}$  and $(\partial_{t_k}\Psi)\Psi^{-1}$ are meromorphic for $z$ in the complex plane with only possible isolate singularities  at $z=0, a_k$.
 Furthermore, using the asymptotic behaviors of $\Psi(z)$  given in \eqref{eq:Psi-infty}--\eqref{eq: Psi-ak},  we see that $U$ and $L$ are of the form appearing in \eqref{eq: L} and \eqref{eq: U}.

 Next, we determine the coefficients in \eqref{eq: L} and \eqref{eq: U}. It follows from the fact $\det \Psi=1$ that $\tr L=\tr U=0$. Hence all the coefficients $\hat{A}_k$, $k=0,\cdots, m$, and $\hat{B}$ are trace-zero.
 Substituting \eqref{eq:Psi-infty} into the first equation of the system \eqref{eq: Lax pair}, we find after comparing the
 coefficient of $\frac{1}{z}$ on both sides of the equation that
\begin{equation}\label{eq:Ak-Psi}
\sum_{k=0}^m\hat{A}_k=\begin{pmatrix}
               n+\frac{\alpha}{2} & t_1 (\Psi_1)_{12}\\
              -t_1(\Psi_1)_{21} & -(n+\frac{\alpha}{2} )              \end{pmatrix},
              \end{equation}
              where $\Psi_1$ is the coefficient in  the large $z$ expansion of $\Psi(z)$ in \eqref{eq:Psi-infty}.
Similarly, inserting the behavior of $\Psi $ at infinity, i.e. \eqref{eq:Psi-infty}, into the second equation of the system \eqref{eq: Lax pair}, we find
\begin{equation}\label{eq:B-Psi}
\hat{B}=\begin{pmatrix}
                0&  t_1(\Psi_1)_{12}\\
             -t_1(\Psi_1)_{21} & 0     \end{pmatrix}.
              \end{equation}
Using the behaviors  of  $\Psi(z)$ as $z\to 0$ and $z\to a_k$, given in \eqref{eq: Psi-origin-1}--\eqref{eq: Psi-ak}, we have
 \begin{equation}\label{eq: AkTrace}
  \det \hat{A}_0=-\alpha^2/4 \quad \mbox {and} \quad  \det \hat{A}_k=0, \quad k=1,\cdots, m.
 \end{equation}
 Hence,  the equations \eqref{eq:Ak-Psi} and \eqref{eq: AkTrace} imply that
  the coefficients $\hat{A}_k$, $ k=1,\cdots, m$, can be parameterized in the form appearing in \eqref{eq:A0} and \eqref{eq:Ak}.
The equation \eqref{eq:B} then follows from  \eqref{eq:A0}, \eqref{eq:Ak}, \eqref{eq:Ak-Psi} and \eqref{eq:B-Psi}.

\begin{remark}
 When $m=2$, the Lax pair \eqref{eq: Lax pair} is equivalent to the Garnier system  \cite[(3.6)]{KNS} appeared in the studies of the degeneration scheme of 4-dimensional Painlev\'e-type equations.
\end{remark}

\subsection{Proof of Theorem \ref{cpv}}
The compatibility condition $\delta\frac{d}{dz}\Psi=\frac{d}{dz}\delta\Psi$ gives the zero-curve equation
   \begin{equation*}
 \delta L(z)-  \frac{d}{dz} U(z)+[L(z),U(z)]=0,
 \end{equation*}
 which is equivalent to
 \begin{equation}\label{eq:zero-curve-2}
 \left\{ \begin{array}{ll}
                  \delta \hat{A}_0=[\hat{B}, \hat{A}_0], \\
                   %  \frac{d}{ds} A_1=[B-\frac{a_1}{2}\sigma_3, A_1], \\
                       \delta \hat{A}_k=[\hat{B}-\frac{t_k}{2}\sigma_3, \hat{A}_k], \quad k=1,\cdots, m.
                    \end{array}
\right.
 \end{equation}
 Substituting \eqref{eq:A0}--\eqref{eq:B} into \eqref{eq:zero-curve-2}, we arrive at
  \begin{equation}\label{eq:zero-curve-3}
 \left\{ \begin{array}{ll}
                 \delta ((a-\frac{\alpha}{2})y)=2ab_1y,\\
                     \delta((a+\frac{\alpha}{2})/y) =2ab_2/y,\\
                       \delta (u_kv_k)=b_1u_kv_k^2+b_2u_k, \quad k=1,\cdots, m,\\
                        \delta (u_ky)=-t_ku_k y+2b_1u_k v_k y, \quad k=1,\cdots, m.
                                           \end{array}
\right. \end{equation}
We obtain from the first two equations that
 \begin{equation}\label{eq:dy}
 \delta \ln y=b_1-b_2=\sum_{k=1}^mu_k(v_k-1)^2-2n-\alpha.
 \end{equation}
The system of nonlinear differential equations  \eqref{m-eq:cpv} then follows from \eqref{eq:dy}
 and the last $2m$ equations of \eqref{eq:zero-curve-3}.

Substituting  into the first equation of the Lax pair the expansion \eqref{eq:Psi-infty},  we obtain after comparing the coefficients of $1/z^2$ in the expansion
\begin{equation}\label{eq:Lexpand2}
-\frac{t_1}{2}[\sigma_3\Psi_1, \Psi_1]-\frac{t_1}{2}[\Psi_2, \sigma_3]+\left(n+\frac{\alpha}{2}\right)[\Psi_1, \sigma_3]-\Psi_1=\sum_{k=1}^ma_k\hat{A}_k.
 \end{equation}
This, together with \eqref{eq:B-Psi} and \eqref{eq:Ak}, implies
\begin{equation}\label{eq:LExpandEntry}
t_1(\Psi_1)_{11}=-t_1^2(\Psi_1)_{12}(\Psi_1)_{21}-\sum_{k=1}^mt_k(\hat{A}_k)_{11}=b_1b_2+\sum_{k=1}^m t_ku_kv_k.
 \end{equation}
From \eqref{def:Y} and \eqref{def:Psi}, we find that
\[p(n,\vec{t})=t_1\left(\Psi_1\right)_{11},\]
which combined with \eqref{sigmap} and \eqref{eq:LExpandEntry} gives us
 \begin{align}
\sigma_n=&t_1\left(\Psi_1\right)_{11}+n(n+\alpha)\label{sgpsi}\\
=&b_1b_2+\sum_{k=1}^m t_ku_kv_k+n(n+\al).\nonumber
\end{align}
Replacing  $b_1$ and $b_2$ by  \eqref{eq:b-k}, in view of \eqref{eq:H-v} , we come to \eqref{eq:H}.
It follows directly from   \eqref{eq:H} that  \eqref{eq:H-equation} are equivalent to   \eqref{m-eq:cpv}.

\subsection{Proof of Theorem \ref{LO-RHP} by Lax pair}
We aim to prove \eqref{Rruv}--\eqref{m-DhR}. Substituting in the Lax pair \eqref{eq: Lax pair} the expansion   \eqref{eq: Psi-ak}, we obtain
 \begin{equation}\label{eq: AkPsi}
 \begin{aligned}
\hat{A}_k=&-\frac{\omega_k}{2\pi i} \Psi^{(k)}(a_k)  \left(
                               \begin{array}{cc}
                                 0 & 1  \\
                                 0 & 0 \\
                                 \end{array}
                             \right)  \Psi^{(k)}(a_k)^{-1}\\
&= -\frac{\omega_k}{2\pi i}\left(
                               \begin{array}{cc}
                                 -\Psi^{(k)}_{11}(a_k)\Psi^{(k)}_{21}(a_k)   &  \Psi^{(k)}_{11}(a_k)^2   \\
                                 - \Psi^{(k)}_{21}(a_k)^2  & \Psi^{(k)}_{11}(a_k)\Psi^{(k)}_{21}(a_k) \\
                                 \end{array}
                             \right),
\end{aligned}
\end{equation}
where $ \Psi^{(k)}(z)$ appears in  \eqref{eq: Psi-ak} for $k=1,\cdots, m$. Using \eqref{def:Y}, \eqref{def:Psi}, \eqref{eq:B-Psi} and  \eqref{eq: AkPsi}, we obtain
\eqref{Rruv}, \eqref{btb1b2} and \eqref{hnt1b1y}. \eqref{uvRr} are direct consequence of \eqref{Rruv}.
Combining \eqref{btb1b2} with \eqref{b1b2} which results from the definitions of $b_1$ and $b_2$, we get \eqref{btrR}.

Now we go ahead with the derivation of \eqref{m-DhR}. Substituting the expansion \eqref{eq:Psi-infty} into the last equation of \eqref{eq: Lax pair} yields
\begin{equation*}
 \partial_{t_k}(t_1\Psi_1)=-\hat{A}_k, \qquad k=2,\cdots, m,
 \end{equation*}
 where $\Psi_1$ is the coefficient in the expansion \eqref{eq:Psi-infty}. Recalling \eqref{eq:Ak}, \eqref{eq:B-Psi} and \eqref{sgpsi}, we arrive at a system of differential equations with respect to $t_k$
 \begin{equation}\label{eq:Dsk}
 \left\{ \begin{array}{ll}
                 \partial_{t_k} \ln (b_1y)=-u_k/b_1,\\
                   \partial_{t_k}\ln(b_1b_2)=u_kv_k^2/b_1-u_k/b_1,    \\
                    \partial_{t_k} \sigma_n=u_kv_k,                                 \end{array}
\right. \end{equation}
for $k=2,\cdots, m$. To study the case $k=1$, we substitute the expansion \eqref{eq:Psi-infty} into the identity
 $$z\frac{d}{dz}\Psi(z)-\delta\Psi(z)=(zL-U)\Psi(z)$$
which comes from the first two equations of \eqref{eq: Lax pair}, and obtain after comparing the coefficients of $1/z$ on both sides of the resulting equality
\begin{equation*}
\delta(t_1\Psi_1)=\left(n+\frac{\alpha}{2}\right)[t_1\Psi_1, \sigma_3]-\sum_{k=1}^mt_k\hat{A}_k.
 \end{equation*}
Then, it is seen  from the above equation
\begin{equation}\label{eq:dH}
 \delta \sigma_n=\delta (t_1(\Psi_1)_{11})=\sum_{k=1}^mt_ku_kv_k,
 \end{equation}
\begin{equation}\label{eq:dby}
\delta(b_1y)=\delta (t_1(\Psi_1)_{12})=-(2n+\alpha)b_1y-\sum_{k=1}^mt_ku_ky,
 \end{equation}
and
\begin{equation}\label{eq:db2y}
\delta(b_2/y)=-\delta (t_1(\Psi_1)_{21})=(2n+\alpha)b_2/y+\sum_{k=1}^mt_ku_k^2v_k/y.
 \end{equation}
Combining \eqref{eq:Dsk} with \eqref{eq:dH}--\eqref{eq:db2y}, we obtain
\begin{equation}\label{eq:Ds1}
 \left\{ \begin{array}{ll}
                 \partial_{t_1}\ln (t_1^{2n+\alpha}b_1y)=-u_1/b_1,\\
                   \partial_{t_1}\ln (b_1b_2)=u_1v_1^2/b_1-u_1/b_1,    \\
                    \partial_{t_1} \sigma_n=u_1v_1.                              \end{array}
\right. \end{equation}
A combination of \eqref{eq:Dsk} and \eqref{eq:Ds1} yields
\begin{equation}\label{eq:Dtk}
 \left\{ \begin{array}{ll}
                 \partial_{t_k} \ln (t_1^{2n+\alpha}b_1y)=-u_k/b_1,\\
                   \partial_{t_k}\ln(b_1b_2)=u_kv_k^2/b_1-u_k/b_1,    \\
                    \partial_{t_k} \sigma_n=u_kv_k,                                 \end{array}
\right. \end{equation}
for $k=1,\cdots, m$, and the first equation of \eqref{eq:Dtk} together with \eqref{hnt1b1y} gives us \eqref{m-DhR}.

It is seen from \eqref{def:Y} that
\begin{equation}\label{eq:alphaY}\alpha_n=(Y_1)_{11}+\frac{(Y_2)_{12}}{(Y_1)_{12}},\end{equation}
where $Y_1$ and $Y_2$ are the coefficients in the expansion of $Y(z)$ as $z\to\infty$; see Deift \cite{Deift}.
Recalling the relation \eqref{def:Psi}, we have
\begin{equation}\label{eq:alphaPsi}\alpha_n=t_1(\Psi_1)_{11}+t_1\frac{(\Psi_2)_{12}}{(\Psi_1)_{12}},\end{equation}
where $\Psi_1$ and $\Psi_2$ are the coefficients in the expansion \eqref{eq:Psi-infty}.
Now, the relation \eqref{alRu} follows directly from \eqref{eq:alphaPsi} and the $(12)$-entry of \eqref{eq:Lexpand2}.

\begin{remark}
With the relations \eqref{Rruv}, \eqref{btb1b2} and the definition \eqref{eq:b-k}, we can show that the ladder operators given by Lemma \ref{lad-oper} are indeed the first column of the Lax pair relating $z$, i.e. the first equation of \eqref{eq: Lax pair}. Note that, to derive the raising operator, we also make use of \eqref{m-S2'1} which can be deduced by using the definitions of $R_{n,k}$ and $r_{n,k}$ together with the fact that $\bt_n=h_n/h_{n-1}$.

In addition to the Lax pair, by using the asymptotic behavior of $\Psi(z)$ at $\infty$ given by \eqref{eq:Psi-infty}, we can establish the following difference equation
\begin{align*}
 \Psi(z;t,n+1)=& \left(
                               \begin{array}{cc}
                                 z+(\Psi_1)_{11}(n+1)-(\Psi_1)_{11}(n) &\qquad -(\Psi_1)_{12}(n)  \\
                                 \qquad(\Psi_1)_{21}(n+1) & 0 \\
                                 \end{array}
                             \right)   \Psi(z;t,n),
\end{align*}
which is equivalent to the three-term recurrence relation \eqref{OP-recu} satisfied by monic orthogonal polynomials.
Moreover, the  compatibility condition of the above difference equation and the first equation of \eqref{eq: Lax pair} lead to the identities $(S_1)$ and $(S_2)$ given in Lemma \ref{s1s2}.
\end{remark}

\begin{remark}
%Besides \eqref{eq: Lax pair}, we have the derivative of $\Psi$ with respect to $s_k$
%\begin{equation}\label{eq: daPsi}
%\partial_{s_k}\Psi(z;\vec{s})=\frac{A_k(s)}{s_k-s_1z} \Psi(z;\vec{s}), \quad k=2,\cdots, m,
%\end{equation}
%While, the derivative of $\Psi$ with respect to $s_1$  follows from \eqref{eq: Lax pair} and \eqref{eq: daPsi}.
According to  \eqref{ruv}, \eqref{btb1b2} and \eqref{hnt1b1y}, the differential equations \eqref{eq:Dtk} are equivalent to
  \begin{equation}\label{eq:Dbeta}
 \left\{ \begin{array}{ll}
                 \partial_{t_k} \ln h_n=-R_{n,k}, \\
                     \partial_{t_k}  \beta_n=\frac{r_{n,k}^2}{R_{n,k}}-\beta_nR_{n,k},  \\
                       \partial_{t_k} \sigma_n=r_{n,k}, \\                              \end{array}
\right. \end{equation}
for $k=1,\cdots,m$.
It is noted that these equations play an important role in deriving the differential equation satisfied by $\sigma_n$, i.e. \eqref{m-sigmaPDE}.
%Taking into account of the fact that $\partial_{s_j}\partial_{s_k}=\partial_{s_k}\partial_{s_j}$, the above equations give
% \begin{equation}\label{eq:Dbeta-2}
% \left\{ \begin{array}{ll}
%                  \partial_{s_j}  R_{n,k}= \partial_{s_k}  R_{n,j},   \\
%                       \partial_{s_j}  r_{n,k}= \partial_{s_k}  r_{n,j},                        \end{array}
%\right. \end{equation}
%where $j,k=1,\cdots,m$.
%By \eqref{eq:Dbeta-2}, the differential equations \eqref{eq:Ricds} are equivalent to the analogs of  Riccati equations obtained before by using the properties of orthogonal polynomials.

Furthermore, with the identities obtained by the Lax pair, we are able to derive the Riccati equations deduced by using the ladder operator approach. Indeed, in view of \eqref{Rruv}, we find from \eqref{eq:zero-curve-3} that
  \begin{equation*}
\delta r_{n,k}=\frac{r_{n,k}^2}{R_{n,k}}+( b_1b_2)R_{n,k}. \end{equation*}
Using \eqref{b1b2} to replace $b_1b_2$, we get the Riccati equation \eqref{m-Ric2}.
 Since  $R_{n,k}=u_k/b_1$, we have
 $$ \delta(R_{n,k})=\frac{1}{b_1y} \delta(u_ky)-\frac{u_ky}{(b_1y)^2} \delta(b_1y).$$
Replacing  the derivatives of  $u_ky$ and $b_1y$ by  the last equation of \eqref{eq:zero-curve-3} and \eqref{eq:dby} respectively,
we come to the Riccati equation \eqref{m-Ric1}.
\end{remark}

\section{Double Scaling Analysis at the Hard Edge}\label{sec:SL}

\subsection{Proof of Proposition \ref{scaledRr}}

Since $r_{n ,k}$ and $R_{n,k}$ are connected with $\sigma_n$ by the relations given in Lemma \ref{rnRnsigma}, we make use of them to derive what we want.

Changing variables $t_k=s_k/(4n)$ and $\vec{t}=\vec{s}/(4n)$ in \eqref{rDsigma} yields
\begin{align}\label{rnsg}
r_{n,k}\left(\frac{\vec{s}}{4n}\right)=4n\cdot\partial_{s_k}\sigma_n.
\end{align}
Dividing its both sides by $4n$ and taking the limit $n\rightarrow\infty$, in view of the definition of $\sigma(\vec{s})$, we get \eqref{m-limRnrnsigma-1}.

Replacing $r_{n,k}(\vec{t})$ by using \eqref{rnsg} and $t_k$ by $s_k/(4n)$ in \eqref{Rbtsigma}, we have
\begin{align}\label{m-Rsigma-1}
R_{n,k}\left(\frac{\vec{s}}{4n}\right)=\frac{1}{2\bt_n}\left(-\partial_{t_k}\bt_n+{\rm sgn}(\omega_k)\sqrt{(\partial_{t_k}\bt_n)^2+64n^2\bt_n\left(\partial_{s_k}\sigma_n\right)^2}\right),
\end{align}
where
\begin{align*}
\bt_n=&\delta\sigma_n-\sigma_n+n(n+\al),\\
\partial_{t_k}\bt_n=&4n\sum_{j=1}^m s_j\cdot\partial_{s_ks_j}^2\sigma_n,
\end{align*}
with $\delta=\sum\limits_{k=1}^m s_k\partial_{s_k}$.
Taking the series expansion in large $n$ of the right hand side of \eqref{m-Rsigma-1}, we find that the leading order term is given by $4\, {\rm sgn}(\omega_k)\sqrt{\left(\partial_{s_k}\sigma_n\right)^2}$, which indicates that
\[R_{n,k}\left(\frac{\vec{s}}{4n}\right)=O(1),\]
as $n\to\infty$.
With this fact, on replacing $r_{n,k}(\vec{t})$ by using \eqref{rnsg} and $t_k$ by $s_k/(4n)$ in \eqref{m-Ric1},
we are led to \eqref{m-limRnrnsigma-2}.

\subsection{Proof of Theorem \ref{m-resultscaled} and Corollary \ref{tau-Rq}}
To derive the expression for $\sigma(\vec{s})$ in terms of $R_{k}(s)$ and $\delta R_k(\vec{s})$, the PDE(s) satisfied by $R_k(\vec{s})$ and $\sigma(\vec{s})$, we make use of the corresponding finite-$n$ results and the relations given by Proposition \ref{scaledRr}.

Using \eqref{m-Ric1} to get rid of $r_{n,k}$ in \eqref{sigmarR}, we come to an expression for $\sigma_n$ in terms of $R_{n,k}$ and $\delta R_{n,k}$. Replacing $R_{n,k}, \delta R_{n,k}$ and $t_k$ by $R_k, \delta R_k$ and $s_k/(4n)$ respectively in the expression, and sending $n$ to $\infty$ on both sides, we obtain \eqref{m-sgRk-lim}.
Substituting $R_k=-4\partial_{s_k}\sigma$ which comes from Proposition \eqref{scaledRr} into \eqref{m-sgRk-lim}, we get \eqref{m-PDEsigma}, the PDE satisfied by $\sigma$.
Replacing $R_{n,k}$ by $R_k$ and $t_k$ by $s_k/(4n)$ in \eqref{m-RPDE}, and taking the series expansion of both sides for large $n$, by comparing the leading coefficients in $n$, we come to \eqref{m-RPDE-lim}, which completes the proof of Theorem \ref{m-resultscaled}.

For $\vec{t}=(t_1,\cdots,t_m)=\frac{\tau}{4n} (\nu_1, \cdots, \nu_m)$, we have
  \begin{equation*}\label{eq:dtD}
\tau\frac{d}{d\tau} \ln D_n(\vec{t})=\delta \ln D_n(\vec{t})=\sigma_n(\vec{t}).
\end{equation*}
Taking derivative once again yields
 \begin{align*}
\left(\tau\frac{d}{d\tau}\right)^2 \ln D_n(\vec{t})=&\delta \sigma_n(\vec{t})=\sum_{k=1}^m t_kr_{n,k}(\vec{t}),
\end{align*}
where the second equality results from \eqref{eq:dH} and \eqref{ruv}. Taking the limit of both sides as $n\rightarrow\infty$, we find
 \begin{align}
\left(\tau\frac{d}{d\tau}\right)^2 \ln\lim\limits_{n\rightarrow\infty} D_n(\vec{t})=&\frac{\tau}{4}\cdot\sum\limits_{k=1}^m\left(\nu_k\lim_{n\to\infty}\frac{r_{n,k}\left(\vec{t}\right)}{n}\right)\nonumber\\
=&-\frac{\tau }{4}\sum_{k=1}^m\nu_k\hat{R}_k(\tau), \label{eq:LimitdH}
\end{align}
 where the second identity comes from \eqref{m-limRnrnsigma-2} and \eqref{defRhat}. We readily see that \eqref{eq:LimitD_n} is a solution of \eqref{eq:LimitdH}. This completes the proof of Corollary \ref{tau-Rq}.

%\begin{remark}
%Since $\vec{t}=\frac{\tau}{4n}\vec{a}$, we have
%\begin{equation}\label{eq:dtD}
%s\frac{d}{ds}( D_n(\vec{t}))=(\delta D_n)(\vec{t})=\sigma_n(\vec{t}).
%\end{equation}
%Taking derivative once again yields
% \begin{equation}\label{eq:ddtD}
%\left(s\frac{d}{ds}\right)^2( D_n(\vec{t}))=\delta \sigma_n(\vec{t}).
%\end{equation}
%This, together with \eqref{eq:dH} and \eqref{eq:LimitD_n}, implies
% \begin{equation}\label{eq:LimitdH}
%-\lim_{n\to\infty}\frac{\sum\limits_{k=1}^ma_kr_{n,k}\left(\vec{t}\right)}{n}=\sum_{k=1}^ma_kq_k^2(s), \end{equation}
%\end{remark}

\section*{Acknowledgments}
Shulin Lyu was supported by National Natural Science Foundation of China under grant number 12101343 and by Shandong Provincial Natural Science Foundation with project number ZR2021QA061. Yang Chen was supported by the Macau Science and Technology Development Fund under grant numbers FDCT 023/2017/A1 and FDCT 0079/2020/A2, and by the University of Macau under grant number MYRG 2018-00125-FST.
Shuai-Xia Xu was partially supported by
National Natural Science Foundation of China under grant numbers 11971492 and 11571376.

\end{document}